\documentclass[12pt]{article}
\usepackage{amsmath}
\usepackage{graphicx,psfrag,epsf}
\usepackage{enumerate}
\usepackage{natbib}
\usepackage{url} 
\usepackage{mathtools}

\newtagform{brackets}{[}{]}
\usetagform{brackets}

\usepackage{setspace}
 \DeclareGraphicsExtensions{.eps, .ps}
\usepackage{soul,color}
\usepackage{amsfonts}
\usepackage{amssymb}
\usepackage{newfloat}
\usepackage{multirow}
\usepackage{arydshln}
\usepackage{comment}
\usepackage{enumitem}
\usepackage{amsthm}
\usepackage[title]{appendix}
\usepackage{lineno}
\usepackage{booktabs}
\usepackage{tikz}
\usepackage{comment}
\usepackage{xr}
\usetikzlibrary{positioning}
\usepackage{subcaption}
\usepackage{authblk}
\usepackage{adjustbox}

\usepackage{xcolor}
\usepackage{hyperref}

\newcommand{\blind}{1}

\usepackage[margin=1in]{geometry}

\newtheorem{theorem}{Theorem}[section]
\newtheorem{proposition}{Proposition}[section] 
\newtheorem{example}{Example}[section] 

\newtheorem{lemma}[theorem]{Lemma}

\newcommand{\indep}{\perp \!\!\! \perp}

\def\E{\mathbb{E}}

\def\P{\mathbb{P}}

\let\E\relax
\DeclareMathOperator\E{\mathrm{E}} 
\DeclareMathOperator{\Var}{\mathrm{Var}}
\DeclareMathOperator{\Cov}{\mathrm{Cov}}
\DeclareMathOperator{\Cor}{\mathrm{Cor}}

\begin{document}

\def\spacingset#1{\renewcommand{\baselinestretch}%
{#1}\small\normalsize} \spacingset{1.5}


\if1\blind
{
  \title{\bf Heritability: a counterfactual perspective}
  \author[a]{Haochen Lei}
    \author[b]{Jieru Shi}
    \author[a]{Hongyuan Cao}
    \author[c]{Qingyuan Zhao\thanks{Corresponding author. Address for
        correspondence: Centre for Mathematical Sciences, Wilberforce
        Road, Cambridge CB3 0WB, United Kingdom. E-mail address:
        qyzhao@statslab.cam.ac.uk.}}
\affil[a]{Department of Statistics, Florida State University}
\affil[b]{Department of Statistical Science, University College
  London}
\affil[c]{Statistical Laboratory, University of
  Cambridge}
  \maketitle
} \fi

\if0\blind
{
  \bigskip
  \bigskip
  \bigskip
  \begin{center}
    {\LARGE\bf Heritability: a counterfactual perspective}
\end{center}
  \medskip
} \fi

\newcommand{\tval}{0} 

\newcommand{\Mysubsection}[1]{%
        \ifnum\tval=0
        \subsection{#1}%
        \else
        \subsection*{#1}%
        \fi
}

\begin{abstract}
Heritability is a central concept in the long-standing debate about
nature versus nurture in biological and social sciences. However, existing
notions of heritability are based on strong assumptions and do not
use explicit causal models. We propose a new, counterfactual
definition of heritability by adopting the potential
outcomes model in causal inference. Our counterfactual heritability
measures the importance of genetic inheritance by the average
magnitude of difference between an individual with their hypothetical
``non-identical twin'' that is exposed to the exact same environment. We
provide bounds on the counterfactual heritability that can, in principle,
be computed from observational data. We then compare counterfactual
heritability and its associated bounds with common notions of
heritability in population-based studies, twin and sibling studies,
and plant breeding experiments. Our results and comparisons highlight
the importance of clarifying the causal structural assumptions and
counterfactual comparisons in reasoning about heritability.
\end{abstract}

\noindent%
\section{Introduction}
Nature versus nurture is a long-standing debate in
biological and social sciences. Central to this discussion is the
concept of heritability, which, roughly speaking, measures the
proportion of phenotypic variation due to genetic variation.
There are several related yet distinct notions of heritability
\citep{jinks1970comparison,jacquard1983heritability,hartl1997principles,zaitlen2012heritability,barry_how_2023}. The
most commonly used notions are the \emph{broad-sense heritability} $H^2$,
which refers to the proportion of phenotypic variance that is
attributable to total genetic variation; the \emph{narrow-sense
        heritability} $h^2$, which refers to the proportion attributable to
additive genetic effects in a statistical genetics model; and the
\emph{SNP-based heritability} $h_{\text{SNP}}^2$, which refers to narrow-sense
heritability based on single-nucleotide polymorphisms (SNP) only. In twin
studies, heritability is typically measured by Falconer's formula
(twice the difference of the correlation between identical twins and
the correlation between non-identical twins). Closely related are
heritability estimates from
sibling regression using the identity-by-descent (IBD) kinship matrix
\citep{visscher2006assumption}. Heritability is also used by plant and
animal breeders to measure the expected response to selection of a
trait. There, heritability is usually defined on an entry-mean basis
(averaging over experimental replications)
\citep{piepho2007computing,schmidt2019heritability}. 

Although existing notions of genetic heritability can shed light on the mechanisms that influence human physiology and behavior, they have two serious drawbacks.
First and foremost, it has been observed since the beginning of the
GWAS era that, for many human traits, there is a large gap between
heritability estimates from twin studies and from unrelated
individuals \citep{manolio2009finding,yang_common_2010}. Many theories
and explanations for this ``missing heritability problem'' have been
proposed \citep{zuk2012mystery,boyle2017expanded,young2019solving},
and this debate has also motivated a flurry of new methods for heritability estimation
\citep[e.g.][]{yang_common_2010,bulik-sullivan_ld_2015,yang_concepts_2017,young2018relatedness}. One
contributing factor to the ``missing heritability problem'' is that
these heritability estimates use different methods and may simply
approximate different population quantities. For example,
the twin-based estimate from the ACE model (see Example \ref{ex:twin}
below) is usually believed to approximate the narrow-sense
heritability $h^2$, and the GWAS-based estimates are thought to
approximate the SNP-based heritability $h_{\text{SNP}}^2$ (this is
basically how $h_{\text{SNP}}^2$ is ``defined''). However, it is not
easy to communicate the subtle differences between different concepts
of heritability to applied scientists and to the public. Fundamentally,
the difficulty is that these notions of heritability are based on
specific statistical models that make different assumptions for
different designs. This means that they are largely ill-defined when
the modeling assumptions are not satisfied. 

This leads us to the second drawback of existing approaches:
although heritability is usually described as the fraction of
phenotypic variation \emph{attributable} to genetic variation and thus
is an inherently causal concept, most existing notions
of heritability only model statistical relations (e.g., associations
and conditional associations) instead of causal relations. A notable
exception is the ACE model for twin studies
\citep{falconer1996introduction,kohler2011social}. This is a special case of linear structural
equation models that were originally introduced by
\citet{wright1920relative} and are routinely used in behavioural
genetics \citep{neale2016openmx} to estimate heritability. However, the ACE model still
suffers from a lack of causal consideration: the genetic and
environmental effects, which are assumed to be independent in the ACE
model, are generally confounded due to parental genetic effects. This
observation has recently regenerated interests in using
family-based designs in human genetics, but mainly for giving causal
interpretations to GWAS results
\citep{kong2018nature,young2022mendelian,veller2024interpreting,young2024genome}
or reducing the bias of Mendelian randomization studies
\citep{brumpton2020avoiding,tudball2022almost}. In sum, although
heritability is often thought as a form of causal attribution, the
causal perspective is largely missing in the mathematical definitions
and models used in the literature.

In this article, we fill this gap with a new, counterfactual definition of
heritability. The key idea is to measure heritability by the average
magnitude of the squared difference between an
individual and its counterfactual---a \emph{hypothetical}
``non-identical twin'' that goes through exactly the same
environment. The precise construction of this counterfactual will
depend on the specific context and more details can be found below. We
rigorously define this natural
conception of heritability using the Neyman-Rubin potential outcomes
model \citep{imbens2015causal} and the nonparametric structural
equation model \citep{pearl2009,richardson2013single} that generalizes
Wright's linear structural equation model. Because our definition of
heritability is inherently counterfactual---belonging to the third layer of Pearl's
ladder of causality
\citep{pearl2018book,bareinboimPearlHierarchyFoundations2022}, it is
generally only partially identifiable from empirical data. To this
end, we provide lower and upper bounds of counterfactual
heritability that can be computed from observational data by estimating the conditional distribution of the phenotype given genetic
and environmental factors. We further compare counterfactual
heritability and its bounds with several existing notions of
heritability.
\section{Counterfactual Heritability}
\subsection{Definition of Counterfactual Heritability}

Let the random variables $G \in \mathcal{G} \subseteq \mathbb{R}^d$
and $Y \in \mathcal{Y} \subseteq \mathbb{R}$ denote, respectively, the
genotype and phenotype of an individual. Suppose one can construct, at
least hypothetically in a thought experiment, the \emph{potential
        outcomes} $Y(g)$ in an experiment that sets $G$ to a fixed value $g
\in \mathcal{G}$. For example, if $Y$ represents the weight of a
potato, then $Y(g)$ is the weight of the potato when the genotype is
set as $g$ by CRISPR gene editing.

Our definition of counterfactual heritability rests on comparing the potential outcomes of two ``non-identical twins'' under
exactly the same environment. More specifically, consider the
following setting:
\begin{equation}
        G \indep G'\indep \{Y(g): g \in \mathcal{G}\}\quad \mbox{and}\quad G'\overset{d}=G. \tag{C1} \label{orth-uncond}
\end{equation}
Here we assume $G$ and $G'$ are two independent and identically
distributed (IID) genotypes that are independent of the potential
outcomes. In this case, we define \emph{counterfactual heritability}
as
\begin{equation}
        \label{eq:def-h2}
        \xi:=\frac{\Var(Y(G)-Y(G'))}{2\Var(Y(G))},
\end{equation}
where $\xi$ means ``explanability''. That is, counterfactual
heritability is the average magnitude of the squared difference between the
potential outcomes of two ``non-identical twins'' (with genotype $G$
and $G'$), normalized by dividing twice the phenotypic variance.

To make this definition more concrete, suppose the potential outcomes
are determined by the structural equation 
\[
Y(g) = f(g, E),~g \in \mathcal{G},
\]
where $f$ is some unknown function and $E$ contains all environmental
factors (possibly unmeasured and including random variations ariging
in the organism itself). Condition [\ref{orth-uncond}]
essentially assumes $G$, $G'$ and $E$ are independent, and the
counterfactual heritability in [\ref{eq:def-h2}] can be written as
\[
\xi = \frac{\Var(f(G,E) - f(G',E))}{2 \Var(f(G,E))}.
\]
It is clear from this expression that the ``non-identical twins'' are
assumed to experience \emph{exactly the same environment} $E$. As an
even more concrete example, if
\[
Y(g) = \beta_g g + \beta_e E + \beta_{ge} g E,
\]
it is not difficult to verify that
\[
\xi = \frac{\beta_g^2 \Var(G) + \beta_{ge}^2 \Var(G) \Var(E)}{\beta_g^2 \Var(G) + \beta_e^2 \Var(E) + \beta_{ge}^2 \Var(G) \Var(E)}.
\]
Thus, the counterfactual heritability $\xi$ includes the variation
from gene-environment interaction measured by $\beta_{ge}^2 \Var(G)
\Var(E)$. This way of ``accounting'' is different from most existing
notions of heritability that only include the variation of the genetic
``main effects'' $\beta_g^2 \Var(G)$ in this simple model.

It is often not realistic to assume that the gene and environment are
independent factors unconditionally. Instead, there may be some
environmental factors and/or background genetic factors $X$ such that,
given $X$, the gene is \emph{unconfounded} with the rest of the
environment. Formally, this unconfoundedness assumption can be
described using the following modification to condition [\ref{orth-uncond}]:
\begin{equation}
        G \indep G'\indep \{Y(g): g \in \mathcal{G}\} \mid X \quad \mbox{and}\quad G'\overset{d}=G \mid X \tag{C2} \label{orth}
\end{equation}
The first part of [\ref{orth}] on the conditional independence of $G$
and potential outcomes of $Y$ is known as ignorability
\citep{imbens2015causal} or no unmeasured confounding
\citep{hernanCausalInferenceWhat2023} in the causal inference
literature. Assuming a nonparametric structural equation model with
independent errors, this can be read off from a causal diagram for
$G$, $Y$, and background factors \citep{pearl2009}.
We will still define the counterfactual heritability as $\xi$ in
equation [\ref{eq:def-h2}], although the distribution of $G'$ is now
different and the variance is also averaged over the distribution of
$X$. One can view [\ref{orth-uncond}] as a special case of
[\ref{orth}] that takes $X$ to be empty. All our theoretical results
below will be stated in the more general case under condition
[\ref{orth}].

Under condition [\ref{orth-uncond}] or [\ref{orth}], $\xi$ is always
bounded between $0$ and $1$. In fact, the counterfactual heritability
in [\ref{eq:def-h2}] can be written as
\begin{equation*} 
       \xi = 1 - \Cor(Y(G), Y(G')),
\end{equation*}
and it can be shown that $\Cor(Y(G), Y(G')) \geq 0$ (see Lemma \ref{lemma:decomp}). This justifies treating $\xi$ as a notion of heritability or variable importance.

Note that the definition of counterfactual heritability depends
on the distribution of $G$. We view this as a feature of this
definition because heritability is a population genetics concept and
thus should depend on the population under investigation. We will
only investigate mathematical properties of $\xi$ in this Section but
will return to its practical interpretation later.


\subsection{Identifiable bounds on counterfactual heritability}

As mentioned in the Introduction, our definition of counterfactual heritability
relies on a hypothetical thought experiment, because no experiment in
the real world can provide \emph{exactly the same
        environment} for two individuals. Thus, $\xi$
should be regarded as an idealized quantity that we should strive to
approximate as accurately as possible.

From a statistical perspective, the difficulty of approximating $\xi$
is that it is not possible to
observe $Y(G)$ and $Y(G')$ at the same time. This is commonly referred
to as the ``fundamental problem of causal inference''
\citep{holland1986statistics}. In fact, in reality we can observe just
one realization of the genotype $G$ and the phenotype $Y$ for any
given individual. This gives rise the so-called \emph{causal
        identification} problem: can we use the distribution of the observed
variables (in our case, $G$, $Y$, and possibly $X$) to determine a
counterfactual quantity (in our case, the counterfactual heritability
$\xi$)?

To make any progress, we first assume that the potential and factual
outcomes are \emph{consistent}  in the sense that $Y(G) = Y$
\citep{hernanCausalInferenceWhat2023}; this is also referred to as the
\emph{stable unit treatment value assumption} (SUTVA)
\citep{rubin1980randomization} or sometimes the \emph{no interference}
assumption in the causal inference literature. By combining this with
[\ref{orth}], it is possible to identify the distribution of $Y(g)$ or
the average treatment effect $\E(Y(g) - Y(g'))$ for any fixed $g, g'$
\citep{rosenbaum1983central}.

However, the counterfactual heritability $\xi$ in [\ref{eq:def-h2}]
clearly depends on the joint distribution of the potential outcomes,
which cannot be identified from empirical data. So exact
identification of $\xi$ is not possible in
general. Nonetheless, it is possible to provide nontrivial bounds on
$\xi$ as we demonstrate below.

Bounding $\xi$ is an instance of optimal transport
problem. Note that bounding $\xi$ is equivalent to bounding $V =
\Var(Y(G) - Y(G'))$, because the denominator in $\xi$ can be directly
identified by $2\Var(Y(G)) = 2 \Var(Y)$. When the genotype $G$ is
binary, the next result provides
tight upper and lower bounds of $V$ using the Fr\'echet-Hoeffding
inequality \citep{frechet1951tableaux,hoeffding1940masstabinvariante}.

\begin{proposition} \label{thm:binary-bound}
        Suppose $\mathcal{G} = \{0,1\}$. Under condition [\ref{orth}], we have 
        \[\xi_l \leq \xi \leq \xi_u,\]
        where
        \begin{align}
                \xi_l &= \frac{\E\left[\left\{F_{G,X}^{-1}(U) -
                        F_{G',X}^{-1}(U)\right\}^2 \right]}{2\Var(Y)}, \label{eq:xi-l}
                \\
                \xi_u &= \frac{\E\left[\left\{F_{G,X}^{-1}(U) -
                        F_{G',X}^{-1}(1 - U)\right\}^2 \right]}{2\Var(Y)}. \tag*{} 
        \end{align}
        In the equations above, $U \sim \mathrm{Unif}[0,1]$ is independent of $(X,G,G')$ and
        $F_{g,x}^{-1}(u) = \inf \{y: F_{g,x}(y) \geq u\}$ is the
        quantile function (the inverse of the conditional distribution
        function $F_{g,x}(y) = \P(Y \leq y \mid G = g, X = x)$) of $Y$ given
        $G$ and $X$. The equality $\xi = \xi_l$ holds if and only if $Y(0)$
        and $Y(1)$ are comonotonic given $X$ (i.e., their conditional rank
        correlation is 1), and
        the equality $\xi = \xi_u$ holds if and only if $Y(0)$ and $Y(1)$ are
        countermonotonic given $X$ (i.e., their conditional rank correlation
        is -1).
\end{proposition}

When $G$ is not binary, $\xi_l$ is still a tight lower bound for
$\xi$, but obtaining tight upper bounds for $V$ (and thus $\xi$) is
still an open problem \citep{wang2015current}. To provide some
intuition, it can be shown (see Lemma \ref{lemma:decomp} in the Appendix) that if $X$ is empty (so
[\ref{orth}] reduces to [\ref{orth-uncond}]), we have
\begin{equation}\label{1-h2}
        1 - \xi = \frac{\Var\{ \E(Y(G) \mid Y(\cdot))\}}{\Var\{Y(G) \}},
\end{equation}
where $Y(\cdot) = \{Y(g): g \in \mathcal{G}\}$ denotes the potential
outcomes ``schedule''. Thus, if there exists a joint distribution on
$Y(\cdot)$ (with constrained marginals given by the observed data)
such that $\E(Y(G)\mid Y(\cdot))$ is a constant, the trivial upper
bound $\xi \leq 1$ is tight; see \cite{xiao2020note} for some
examples. Heuristically, if we do not have any environmental factor
and observe different phenotypes for two ``non-identical twins'', we
cannot rule out the possibility that such difference is entirely
due to the difference in their genotypes. 


The next result provides bounds in the non-binary case. All
mathematical proofs can be found in the Appendix.

\begin{theorem} \label{thm:main-bound}
        Under [\ref{orth}], 
        we have 
       \[
       \xi_l' \leq \xi_{l}\le \xi \le \min\{\xi_u, \xi_u'\},
       \]
        where
        \begin{align}
                \xi'_{l} &= \frac{\E[\Var\{\E(Y \mid G,X)\mid X\}]}{\Var(Y)}= 1 -
                \frac{\E(\Var(Y \mid G,X))+\Var(\E(Y\mid X))}{\Var(Y)}, \label{eq:xi_lower}\\
                \xi_u'&=\frac{\E(\Var(Y \mid X))}{\Var(Y)} = 1-\frac{\Var(\E(Y \mid X))}{\Var(Y)},\label{eq:xi_upper}
        \end{align}
        and the definitions of $\xi_l$ and $\xi_u$ can be found in
        Proposition
        \ref{thm:binary-bound}.
        The equality $\xi_{l} = \xi$ holds if and only if all
        potential outcomes are
        comonotonic given $X$, and the equality $\xi_l' = \xi_l$ holds if the quantile difference $F^{-1}_{G,X}(u)-F^{-1}_{G',X}(u)$ does not depend on $u$.
\end{theorem}

Theorem \ref{thm:main-bound} provides two lower bounds and tow upper
bounds on the counterfactual heritability $\xi$ that can be computed
using a sample of $(X,G,Y)$. The lower bound $\xi_l$ is tight but
requires modeling the conditional quantile function of $Y$ given $X$
and $G$. The more relaxed lower bound $\xi_{l}'$ could be easier to
compute. On the
other side, the new upper bound
$\xi_u'$ uses the expected conditional variance of
$Y$ given the observed environmental factors. Thus, including some
environmental factors $X$ that are independent of the genes but
predictive of the phenotype can lead to non-trivial upper bounds of
the counterfactual heritability. Intuitively, the genotypes can explain at
most the variance not possibly explained by $X$, so
\[
\frac{1}{2} \Var(Y(G) - Y(G')) \leq \E(\Var(Y \mid X)).
\]
The inequality $\xi \leq \xi_u$ in Proposition \ref{thm:binary-bound}
still holds for non-binary $G$, but it is generally not very
useful. For example, when $\P(G = g) = 1/3$ and $Y(g) \sim
\text{N}(0,1)$ for $g \in \{0,1,2\}$, it is not difficult to show that
$\xi_u = 4/3 > 1$ and thus the upper bound $\xi_u$ is trivial.
\section{Connection with different notions of heritability}
\subsection{Comparison with broad-sense and narrow-sense
  heritability}
Next, we compare counterfactual heritability $\xi$ with other existing notions of heritability.
Roughly speaking, ``broad-sense heritability'', often denoted as
$H^2$, is the proportion of phenotypic variation statistically
explained by \emph{total genetic variation}, including dominance and
epistasis. On the other hand, ``narrow-sense heritability'', often
denoted as $h^2$, refers to phenotypic variation explained by \emph{additive
        genetic variation} only. These quantities are not always precisely defined in the literature (their definition usually requires a specific statistical model). Consider the example below.

\begin{example} \label{ex:one-gene}
        To illustrate the difference between $H^2$ and $h^2$, consider a simple example with just one gene. Assume that the counterfactual model for the phenotype $Y$ is
        $Y(g)=f(g)+E$, where $g\in \{AA,Aa,aa\}$ is the genotype and $E$ is a random environmental factor. Suppose
        \begin{equation} \label{eq:one-gene-f}
                f(AA)=m+a,f(Aa)=m+d,f(aa)=m-a,
        \end{equation}
        where $a$ is the additive/linear effect and $d$ is the dominance
        effect \citep{hartl1997principles}. Suppose $\P(G = AA) = p^2$, $\P(G =
        Aa) = 2pq$, and $\P(G = aa) = q^2$ for some $p + q = 1$, and $E$ is
        independent of $G$ with mean $0$. Then the variance of the genetic
        effect $f(G)$ can be written as
        \begin{align*} 
                \underbrace{\Var(f(G))}_{\sigma_g^2} &= \underbrace{2pq[a+(q-p)d]^2}_{\sigma_a^2}+\underbrace{(2pqd)^2}_{\sigma_d^2},
        \end{align*}
        where $\sigma_a^2$ is the additive effect, and $\sigma_d^2$ is the dominance effect.
        The least-squares projection of $f(G)$ to linear functions of $G$
        (under the coding $AA \mapsto 1$, $Aa \mapsto 0$, $aa \mapsto -1$, so
        $\E(G) = p-q$ and $\Var(G) = 2pq$) is given by $G \beta_0$ where
        $\beta_0 = a + (q-p)d$. This justifies calling the first term on the
        right hand side the variance of the additive/linear genetic
        effect because
        $\sigma_a^2 = \Var(G \beta_0)$. If we denote $\Var(E) = \sigma_e^2$ so
        $\Var(Y) = \sigma_g^2 + \sigma_e^2$, the broad-sense and
        narrow-sense heritability are defined as
        \[
        H^2 = 
        \frac{\sigma_a^2 + \sigma_d^2}{\sigma_a^2 + \sigma_d^2 + \sigma_e^2}~\text{and}~
        h^2 
        = \frac{\sigma_a^2 }{\sigma_a^2 + \sigma_d^2 + \sigma_e^2}.
        \]
        In this settings, the counterfactual heritability
        \[
        \xi = \frac{\Var(Y(G) - Y(G'))}{2 \Var(Y(G))} = \frac{\Var(f(G))}{\Var(Y)} = H^2
        \]
        is the same as the broad-sense heritability. When there are multiple causal genes, the equality $\xi = H^2$ still
holds as long as there is no gene-environment interaction.
When there is gene-environment interaction, counterfactual
heritability will include $G \times E$ interaction as part of its
numerator as discussed earlier.
\end{example}

Next, we compare $H^2$ and $h^2$ with counterfactual heritability $\xi$ and
its bounds in Theorem \ref{thm:main-bound}. This is not
straightforward because the precise definitions of $H^2$ and $h^2$ are
often tied to specific parametric models of the phenotype. To make
these quantities comparable to $\xi$ when the
parametric models are not correct, we will use the following nonparametric definition
\begin{equation}
\label{def:H2}
   H^2 = \frac{\Var(\E(Y \mid G))}{\Var(Y)}
\end{equation}
and
\begin{equation}
\label{def:h2}
   h^2 = \frac{\Var(G^T \beta_0)}{\Var(Y)},
\end{equation}
where $G^T \beta_0$ is the least-squares projection of $\E(Y \mid G)$ onto linear functions of $G$, that is, $$\beta_0 = \arg\min_{\beta \in \mathbb{R}^d} \Var(\E(Y \mid G) - G^T \beta).$$ It follows that $0 \leq h^2 \leq H^2 \leq 1$. Here we deliberately exclude any measured environmental factors $X$ even though they are often used in practical estimation methods. This is because the exact regression model often changes the definition of $H^2$ and $h^2$ and makes the comparison more difficult.

It can be shown that $H^2\le \xi_l'$ if $G \indep X$ and $H^2 = \xi_l'$ if $(Y,G) \indep X$
(see Lemma \ref{lem:xil-h2}). However,
$H^2\le \xi_l'$ is not always true; for example, when $G = X$ with
probability 1, we have $\xi_l' = 0$ but $H^2$ can still be positive.

\begin{table}[t]
        \centering
        \caption{A comparison of different notions of population heritability in different counterfactual models of a phenotype $Y$.}
        \label{tab:heritability1}
    \begin{adjustbox}{max width=\textwidth}
        \begin{tabular}{l cccccc}
                \toprule
                $Y(g)=$ & Narrow $h^2$ & Broad $H^2$ & Counterfactual $\xi$ & Lower $\xi_l'$ & Tight lower $\xi_l$ & Upper $\xi_u'$ \\
                \midrule
                $\beta_1 g_1+\beta_2g_2+E_1$ & $\frac{\beta_1^2+\beta_2^2}{\beta_1^2+\beta_2^2+1}$ & $\frac{\beta_1^2+\beta_2^2}{\beta_1^2+\beta_2^2+1}$& $\frac{\beta_1^2+\beta_2^2}{\beta_1^2+\beta_2^2+1}$&
                $\frac{\beta_1^2+\beta_2^2}{\beta_1^2+\beta_2^2+1}$&
                $\frac{\beta_1^2+\beta_2^2}{\beta_1^2+\beta_2^2+1}$&1 \\
                $\beta g_1g_2+E_1$ & $\frac{2\beta^2}{4+3\beta^2}$ & $\frac{3\beta^2}{4+3\beta^2}$ & $\frac{3\beta^2}{4+3\beta^2}$& $\frac{3\beta^2}{4+3\beta^2}$ &$\frac{3\beta^2}{4+3\beta^2}$ &$1$ \\
                $\beta g_1 E_2+E_1$ & $0$ & $0$ & $\frac{\beta^2}{4+2\beta^2}$&$0$&$\frac{(\sqrt{\beta^2+1}-1)^2}{4+2\beta^2}$&$1$ \\
                $\beta_1g_1+\beta_2 X+E_1$ & $\frac{\beta_1^2}{\beta_1^2+\beta_2^2+1}$ & $\frac{\beta_1^2}{\beta_1^2+\beta_2^2+1}$ & $\frac{\beta_1^2}{\beta_1^2+\beta_2^2+1}$ & $\frac{\beta_1^2}{\beta_1^2+\beta_2^2+1}$ & $\frac{\beta_1^2}{\beta_1^2+\beta_2^2+1}$ & $\frac{1+\beta_1^2}{\beta_1^2+\beta_2^2+1}$\\
                $\beta g_1 X+E_1$ & $0$ & $0$ & $\frac{\beta^2}{4+2\beta^2}$&$\frac{\beta^2}{4+2\beta^2}$&$\frac{\beta^2}{4+2\beta^2}$&$\frac{4+\beta^2}{4+2\beta^2}$ \\
                \bottomrule
        \end{tabular}
    \end{adjustbox}
\end{table}


Table \ref{tab:heritability1} compares counterfactual
heritability and its bounds with narrow-sense and broad-sense
heritability in a number of examples. To simplify the
expressions, we consider two independent, binary genotypes $G_1,G_2
\sim \text{Bern}(0.5)$ and some IID normally distributed observed
environmental factor $X$ and unobserved factors $E_1,E_2$ with mean
$0$ and variance $1/4$. Table \ref{tab:heritability2} further provides
the numeric values of
different notions of heritability in a variety of examples
with two causal genetic variants $G_1,G_2$ (the exact settings are
described in the Online Supplement).

It is easy to see from Table \ref{tab:heritability1} that all six
notions of heritability are different. As another instructive example,
in the two settings where $Y(g) = \beta g_1 E_2 + E_1$ and $Y(g) =
\beta g_1 X + E_1$, the counterfactual heritability $\xi$ is the same
but the interval $[\xi_l, \xi_u']$ is tighter in the latter case
because the environmental factor $X$ is known.

In Table \ref{tab:heritability2}, we observe that when there is no
interaction between gene and unknown environmental factors in
generating $Y$ (as in the examples (1),(2),(4),(5)), the true
counterfactual heritability $\xi$, its tight lower bound $\xi_l'$, and
its lower bound $\xi_l$ are the same. This is because the outcomes
$Y(G)$ and $Y(G')$ are comonotonic. Additionally, we observe that the
narrow-sense heritability $h^2$ and broad-sense heritability $H^2$ are
both no greater than $\xi$. This is because $h^2$ and $H^2$ only
consider genetic effect as an independent component in determining the
outcome $Y$, and do not account for the interaction effect between
gene and environment. Furthermore, when the known covariate $X$ is not
included in the model, the upper bound $\xi_u'$ becomes trivial
($\xi_u'=1$).

\begin{table}[t]
	\centering
	\caption{A numerical comparison of different notions of
		heritability in a population. In this table, $X$ is an
		environmental factor that is independent of $G$, 
		and the genetic copy is derived from the population.}
	\label{tab:heritability2}
	\begin{adjustbox}{max width=\textwidth}
		\begin{tabular}{l cccccc}
			\toprule
			$Y(g)=$& Narrow $h^2$ & Broad $H^2$ & Counterfactual $\xi$ & Lower $\xi_l'$ & Tight lower $\xi_l$ & Upper $\xi_u'$ \\\hline
			$(1) \quad g_1+g_2+E_1$ & $0.67$ & $0.67$ & $0.67$ &$0.67$&$0.67$& $1$\\\hline
			$(2) \quad g_1g_2+E_1$ & $0.57$ & $0.71$ & $0.71$ &$0.71$&$0.71$& $1$\\\hline
			$(3) \quad g_1E_2+E_1$      & $0.36$    & $0.36$    & $0.45$ &$0.36$&$0.39$& $1$ \\\hline
			$(4)\quad g_1+X+E_1$&  $0.5$& $0.5$& $0.5$ & $0.5$ & $0.5$ & $0.75$\\\hline
			$(5)\quad g_1X+E_1$ &  $0.36$ & $0.36$ & $0.45$ & $0.45$ & $0.45$ & $0.75$\\\hline
			$(6)\quad 1(g_1+g_2+E_1>0)$        & $0.14$ & $0.30$ & $0.85$ &$0.30$&$0.85$&$1$\\\hline
			$(7)\quad 1(g_1g_2+E_1>0)$         & $0.19$ & $0.30$ & $0.66$ &$0.30$& $0.66$&$1$\\\hline
			$(8)\quad 1(g_1E_2+E_1>0)$  & $0.16$    & $0.19$ & $0.52$ &$0.19$& $0.52$ & $1$ \\\hline
			$(9)\quad 1(g_1+X+E_1>0)$ & $0.06$& $0.08$ & $0.67$ & $0.28$ & $0.67$ & $0.95$\\\hline
			$(10)\quad 1(g_1X+E_1>0)$ & $0.16$ & $0.19$& $0.52$& $0.52$& $0.52$ & $0.95$\\\hline
		\end{tabular}
	\end{adjustbox}
\end{table}

\subsection{Counterfactual heritability of one generation and
  comparison with twin-study heritability}
\label{sec:one-generation}

Our definition of counterfactual heritability requires a suitable
choice of environmental confounders $X$, so $\xi$ can be understood as
the proportion of phenotypic variation that is attributable (in a
counterfactual sense) to genetic variation \emph{given} $X$. In a
population that evolves naturally, such as humans, it is reasonable to
let $X$ collect all haplotypes of the parents inherited from grandparents, so $\xi$ essentially
represents the counterfactual heritability in \emph{one generation}
due to the randomness of meiosis.


The next example demonstrates, in a setting similar to Example
\ref{ex:one-gene}, a close connection between this notion of
one-generation counterfactual heritability and twin-study
heritability that is often calculated by Falconer's
formula,
\begin{equation}
\label{def:twin}
    h^2_{\text{twin}} = 2(\rho_{\text{MZ}}-\rho_{\text{DZ}}),
\end{equation}
where $\rho_{\text{MZ}}$ and $\rho_{\text{DZ}}$ are the correlations
between monozygotic and dizygotic twins.

\begin{example} \label{ex:twin}
        Consider $g \in \{AA, Aa, aa\}$ and the function $f$ in
        equation [\ref{eq:one-gene-f}].
        Suppose $Y_1(g)$ and $Y_2(g)$ are potential outcomes of two siblings in a family given by
        \[
        Y_1(g) = f(g) + F + E_1, ~ Y_2(g) = f(g) + F + E_2,
        \]
        where $F$ is a family ``fixed" effect (not necessarily observed) and $E_1, E_2$ are IID idiosyncratic noise variables. Let $G$ and $G'$ be two IID genotypes that are randomly determined given maternal genotype $G_M$ and paternal genotype $G_P$ according to Mendelian laws. If the siblings are monozygotic twins, their outcomes can be modeled by
        \[
        Y_1 = Y_1(G) = f(G) + F + E_1,~ Y_2 = Y_2(G) = f(G) + F + E_2.
        \]
        If the siblings are dizygotic twins, their outcomes can be modeled by
        \[
        Y_1 = Y_1(G) = f(G) + F + E_1,~ Y_2 = Y_2(G') = f(G') + F + E_2.
        \]
        This is a special case of the ACE model in the literature
        \citep{neale1992methodology}, which contains (additive) genetic, common
        environment, and random environment effects, with just a single causal
        gene. It is shown in Example \ref{ex:one-gene} that $\Var(f(G)) =
        \sigma_a^2 + \sigma_d^2$. By enumerating all possibilies of $X
        = (G_M,G_P)$, it can be shown that $\Cov(f(G), f(G')) = \frac{1}{2}
        \sigma_a^2 + \frac{1}{4} \sigma_d^2$ \citep[page
        436]{hartl1997principles}. Thus, the one-generation counterfactual
        heritability is given by
        \begin{align*}
                \xi &= \frac{\Var(Y_1(G)-Y_1(G'))}{2\Var(Y_1(G))}\\&=\frac{\Var(f(G))-\Cov(f(G),f(G'))}{\Var(Y_1(G))} \\&= \frac{\frac{1}{2} \sigma_a^2 + \frac{3}{4} \sigma_d^2}{\Var(Y_1)}.
        \end{align*}
        The twin-study heritability $h^2_{\text{twin}}$ can be computed using
        \begin{align*}
                \rho_{\text{MZ}}&=\frac{\Cov(Y_1(G),Y_2(G))}{\Var(Y_1(G))}= \frac{\Var(f(G)) + \sigma_F^2}{\Var(Y_1)}, \\
                \rho_{\text{DZ}}&=\frac{\Cov(Y_1(G),Y_2(G'))}{\Var(Y_1(G))} = \frac{\Cov(f(G), f(G')) + \sigma_F^2}{\Var(Y_1)},
        \end{align*}
        where $\sigma_F^2 = \Var(F)$ is the variance of the family fixed effect.
        From here, it is easy to show that 
       \[
       h^2_{\text{twin}} = 2 \xi.
       \]
        That is, our one-generation counterfactual heritability is exactly a
        half of the commonly used twin-study heritability in this simple ACE
        model. In fact, it is easy to see that the equality $h^2_{\text{twin}}
        = 2 \xi$ extends to polygenic traits (multivariate $G$) as long as
        there is no gene-environment interaction. It can also be shown
        (see Lemma \ref{lemma:twin} in the Online Supplement) that
        $h^2_{\text{twin}}=2\xi_l'$ if $Y_1(G)\indep Y_2(G')\mid X$
        (i.e., there is no shared environmental influence besides $X$).
\end{example}

Example \ref{ex:twin} shows that we can understand Falconer's formula
as (twice) the one-generation counterfactual heritability in twin
studies under an additive model. In this sense, counterfactual
heritability provides a natural generalization of Falconer's formula
to other study designs and non-additive models.

\begin{figure}[t]
        \centering
        \begin{subfigure}[b]{0.45\columnwidth} \centering
                \begin{tikzpicture}
                        \node (X) {$X$};
                        \node (Z) [right of=X] {$G$};
                        \node (Y) [right of=Z] {$Y$};

                        \draw [->] (X) -- (Z);
                        \draw [->] (Z) -- (Y);
                        \draw [->, bend left = 60] (X) to (Y);
                \end{tikzpicture}
                \caption{A causal graph.}
                \label{fig:direct-indirect-a}
        \end{subfigure}
        \begin{subfigure}[b]{0.45\columnwidth} \centering
                \begin{tikzpicture}
                        \node (X) {$X$};
                        \node (Z) [right=0.4cm of X] {$G$};
                        \node (Y) [right=0.4cm of Z] {$Y(X, G)$};

                        \draw [->] (X) -- (Z);
                        \draw [->] (Z) -- (Y);
                        \draw [->, bend left = 45] (X) to (Y);
                \end{tikzpicture}
                \caption{Potential outcomes.}
                \label{fig:direct-indirect-b}
        \end{subfigure}

        \begin{subfigure}[b]{0.45\columnwidth} \centering
                \begin{tikzpicture}
                        \node (X) {$X$};
                        \node (Z) [right=0.4cm of X] {$G'$};
                        \node (Y) [right=0.4cm of Z] {$Y(X, G')$};

                        \draw [->] (X) -- (Z);
                        \draw [->] (Z) -- (Y);
                        \draw [->, bend left = 45] (X) to (Y);
                \end{tikzpicture}
                \caption{Fraternal counterfactual.}
                \label{fig:direct-indirect-c}
        \end{subfigure}
        \begin{subfigure}[b]{0.45\columnwidth} \centering
                \begin{tikzpicture}
                        \node (X) {$X'$};
                        \node (Z) [right=0.4cm of X] {$G''$};
                        \node (Y) [right=0.4cm of Z] {$Y(X', G'')$};

                        \draw [->] (X) -- (Z);
                        \draw [->] (Z) -- (Y);
                        \draw [->, bend left = 45] (X) to (Y);
                \end{tikzpicture}
                \caption{Unrelated counterfactual.}
                \label{fig:direct-indirect-d}
        \end{subfigure}
        \begin{subfigure}[b]{0.45\columnwidth} \centering
                \begin{tikzpicture}
                        \node (X) {$X$};
                        \node (Z) [right=0.4cm of X] {$G''$};
                        \node (Y) [right=0.4cm of Z] {$Y(X, G'')$};
                        \draw [->] (Z) -- (Y);
                        \draw [->, bend left = 45] (X) to (Y);
                \end{tikzpicture}
                \caption{Adopted counterfactual.}
                \label{fig:direct-indirect-e}
        \end{subfigure}
        \caption{An illustration of different notions of counterfactuals
                and notions of counterfactual heritability. Panel (a) shows a
                simple graphical model for direct and indirect genetic effects
                ($X$ is parental haplotype, $G$ is offspring genotype, $Y$ is
                offspring phenotype). Panel (b) shows the same graph with the
                phenotype $Y$ denoted by its potential outcome $Y(X,G)$ (they
                are equal by the consistency assumption). The one-generation $\xi$
                compares (b) with a fraternal counterfactual generated by panel (c),
                while the population $\xi$ compares (b) with a unrelated
                counterfactual generated according to (d). The notion of
                heritability used in RDR compares (b) with an adopted
                counterfactual generated according to (e).}
        \label{fig:direct-indirect}
\end{figure}
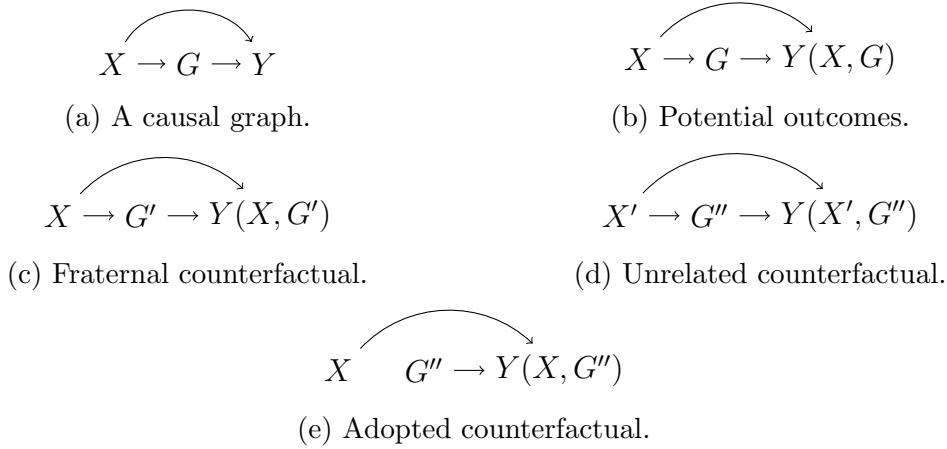

\subsection{Distinguishing notions of heritability by the
        counterfactual comparison}

To further illustrate the concept of counterfactual heritability of
one generation, it is useful to consider the simple graphical model in
Figure \ref{fig:direct-indirect-a} which shows that parental haplotype
$X$ has a direct and indirect effect (via offspring genotypes) on the
offspring phenotype $Y$. The unconfoundedness assumption, that is $G
\indep \{Y(g): g \in \mathcal{G}\} \mid X$ in [\ref{orth}], holds
under a standard interpretation of Figure \ref{fig:direct-indirect} as a
nonparametric structural equation model with independent errors \citep{pearl2009}. Let $G'$ be
an independent realization of $G$ given $X$, $X'$ be an independent
realization of $X$, and $G''$ be an independent realization given
$X'$. In this case, the one-generation
counterfactual heritability is
\[
\xi_{\text{fraternal}} = \frac{\Var\{Y(X,G)-Y(X,G')\}}{2\Var(Y(X,G))},
\]
and the population counterfactual heritability is
\[
\xi_{\text{unrelated}} = \frac{\Var\{Y(X,G)-Y(X',G'')\}}{2\Var(Y(X,G))}.
\]
The data generating mechanisms for the counterfactuals in these
definitions are shown in panels (b-d) of Figure
\ref{fig:direct-indirect}. The one-generation
$\xi_{\text{fraternal}}$ compares an
individual in panel (b) of Figure \ref{fig:direct-indirect} to their
within-family fraternal counterfactual
in panel (c), while the population-based $\xi_{\text{unrelated}}$
compares an individual in
panel (b) to their population (unrelated) counterfactual in panel (d).
In reality, because the haplotype $X$ of the parents is further inherited from their parents, and so on, it is perhaps more appropriate to think about $\xi_{\text{unrelated}}$ as \emph{``infinite-generation'' heritability}.

Table \ref{tab:heritability3} provides a numeric comparison of different notions of heritability in a specific with-in family example.  Since $h^2$ and $H^2$ do not
account for family structure, their values remain the same as in Table
\ref{tab:heritability2}. For models without gene-environment
interaction (as seen in rows analogous to those in Table
\ref{tab:heritability2}), the true counterfactual heritability $\xi$,
the tight lower bound $\xi'_l$, and the lower bound $\xi_l$ are again
identical. Due to the additional information of parents' genotypes
coded in $X$, the upper bound $\xi'_u$ is strictly less than
$1$. Additionally, since now $G$ and $G'$ are conditionally
independent given the parents' genotype, the variance of
$Y(G) - Y(G')$ is smaller than the marginally independent case
in Table \ref{tab:heritability2}. Therefore, the values of $\xi$ in
Table \ref{tab:heritability3} are smaller than in Table
\ref{tab:heritability2}. The exact settings that generated Table
\ref{tab:heritability3} can be found in the Online Supplement.

\begin{table}[t]
        \centering
        \caption{A numerical comparison of different notions of
                heritability in a within-family setting. In this table $X$ is the the parents' average genotype corresponding to $G_1$.}
\label{tab:heritability3}
\begin{adjustbox}{max width=\textwidth}
\begin{tabular}{l ccccccc}
        \toprule
        $Y(g)=$ & Narrow $h^2$ & Broad $H^2$ & One-generation $\xi_{\text{fraternal}}$& Twin $h^2_{\text{twin}}$ & Lower $\xi_l'$ & Tight lower $\xi_l$ & Upper $\xi_u'$ \\\hline
        $(1)\quad g_1+g_2+E_1$ & 0.67 & 0.67 & 0.33 & 0.67 & 0.33 & 0.33 & 0.67 \\\hline
        $(2)\quad g_1g_2+E_1$ & 0.57 & 0.71 & 0.39 & 0.79 & 0.39 & 0.39 & 0.68 \\\hline
        $(3)\quad g_1E_2+E_1$ & 0.36 & 0.36 & 0.23 & 0.36 & 0.18 & 0.20 & 0.82 \\\hline
        $(4)\quad g_1+X+E_1$ & 0.64 & 0.64 & 0.14 & 0.29 & 0.14 & 0.14 & 0.43 \\\hline
        $(5)\quad g_1X+E_1$ & 0.60 & 0.63 & 0.15 & 0.30 & 0.15 & 0.15 & 0.42 \\\hline
        $(6)\quad 1(g_1+g_2+E_1>0)$ & 0.14 & 0.30 & 0.58 & 0.38 & 0.19 & 0.58 & 0.90 \\\hline
        $(7)\quad 1(g_1g_2+E_1>0)$ & 0.19 & 0.30 & 0.41 & 0.35 & 0.18 & 0.41 & 0.88 \\\hline
        $(8)\quad 1(g_1E_2+E_1>0)$ & 0.16 & 0.19 & 0.31 & 0.21 & 0.10 & 0.31 & 0.91 \\\hline
        $(9)\quad 1(g_1+X+E_1>0)$ & 0.14 & 0.20 & 0.27 & 0.11 & 0.05 & 0.27 & 0.76 \\\hline
        $(10)\quad 1(g_1X+E_1>0)$ & 0.21 & 0.25 & 0.28 & 0.19 & 0.10 & 0.28 & 0.82 \\\hline
\end{tabular}
\end{adjustbox}
\end{table}

\subsection{Comparison with relatedness disequilibrium regression}

The relatedness disequilibrium regression (RDR) is a recent approach
that defines heritability through random segregation and thus bears
conceptual similarity to counterfactual heritability
\citep{young2018relatedness}. A key observation in
\cite{young2018relatedness} is that traditional heritability estimator
can be biased by confounding environmental factors. In particular,
even though traditional kinship methods based on twins and siblings
are robust to additive family effects, they can still be biased by
sibling interference effects, meaning one's genotype may influence
the phenotype of their sibling (this is called ``reciprocal genetic
effects'' in \cite{young2018relatedness}). Consider the following
example.

\begin{example} \label{ex:rdr}
Consider a family with two children. Let $Y_i$ be the
phenotype and $g_i$ be the genotype for child $i = 1, 2$. Suppose the
potential outcomes of these two children are given by
\begin{align*}
        Y_1(g_1,g_2) &= f_{\text{direct}}(g_1) + f_{\text{indirect}}(g_2) + F + E_1, \\
        Y_2(g_1,g_2) &= f_{\text{indirect}}(g_1) + f_{\text{direct}}(g_2) +
        F + E_2,
\end{align*}
where $f_{\text{direct}}$ is the ``direct'' genetic effect of one's
genotype on themself, $f_{\text{indirect}}$ is the ``indirect''
genetic effect of one's genotype on their sibling,  $F$ is a family
fixed effect (not observed), and $E_1$, $E_2$ are IID idiosyncratic
noise variables. By computing the covariance of $Y_1$ and $Y_2$ in
terms of the covariance matrix of $(G_1,G_2)$, it is shown in
\citet[Supplementary Note, Appendix A.1]{young2018relatedness} that
Sib-Regression tends to
overestimate the following notion of genetic heritability:
\[
h^2_{\text{RDR}}=\frac{\Var(f_{\text{direct}}(G_1))}{\Var(Y_1)}.
\]
It is further shown that, by using genetic relatedness between
pairs of individuals and their parents, RDR can consistently estimate
this quantity \citep[Supplementary Note, Section
1.4]{young2018relatedness}. In this example, the one-generation
counterfactual heritability proposed here is given by
\[
\begin{split}
        \xi_{\text{fraternal}} &=\frac{\Var(Y_1(G_1,G_2)-Y_1(G_1',G_2))}{2\Var(Y_1)} \\
        &=\frac{\E\{\Var(f_{\text{direct}}(G_1)\mid X)\}}{\Var(Y_1)},
\end{split}
\]
where $G_1'$ is an IID copy of $G_1$ given $X$ ($G_2$ is held fixed in
this comparison because the genotype of child 2 is part of the
environmental factors for child 1). Thus,
$\xi_{\text{fraternal}}$ is different from $h^2_{\text{RDR}}$, and we
always have $\xi_{\text{fraternal}} \leq h^2_{\text{RDR}}$.
\end{example}

The definition of $h^2_{\text{RDR}}$ in this example motivates a
counterfactual quantity
\[
\xi_{\text{adopted}} = \frac{\Var\{Y(X,G)-Y(X,G'')\}}{2 \Var(Y(X,G))},
\]
where $G''$ is the genotype of the offspring in a unrelated family
with parental haplotype $X'$ independent of $X$; see panels (d) and
(e) in Figure \ref{fig:direct-indirect}. So in the definition of
$\xi_{\text{adopted}}$, an individual is compared with their adopted
counterfactual. It is not difficult to show that $h^2_{\text{RDR}} =
\xi_{\text{adopted}}$ in the setting of Example \ref{ex:rdr}. However,
even though it may be sensible to compare to an adopted counterfactual
in some scenarios, this definition seems to be flawed: in
the setting of Example \ref{ex:rdr},
$h^2_{\text{RDR}}$ or $\xi_{\text{adopted}}$ can be larger than $1$
when $f_{\text{direct}}(G_1)$ and $f_{\text{indirect}}(G_2)$ are
negatively correlated.

\subsection{Comparison with plant breeding heritability}
In the plant breeding literature, heritability is usually defined on
an entry-mean basis because the phenotype of interest is usually an
aggregated value over replicated genotypes in experiments
\citep{piepho2007computing,schmidt2019heritability}. The definition of
counterfactual heritability can be easily extended to this case by
considering the counterfactual of the aggregated phenotype, as
illustrated below.

Let $Y(g, x)$ denote the potential outcome of a phenotype of the plant
when the genotype is set to $g \in \{1,\dots,n_g\}$ and environment is set
to $x \in \{1,\dots,n_x\}$. A plant
breeding experiment often measures $Y_i = Y(G_i, X_i), i=1,\dots,n$
under a factorial design such that
\[
\sum_{i=1}^n 1_{\{G_i = g, X_i = e\}} = n_r,~\text{for all}~g~\text{and}~x,
\]
where $n_r$ is the number of replications (so $n = n_g n_x
n_r$). Denote the
mean phenotype for all plants with genotype $g = 1,\dots,n_g$ and
environment $x = 1,\dots,n_x$ as
\[
\bar{Y}(g, x) = \frac{1}{n_r} \sum_{i=1}^n 1_{\{G_i = g, X_i = x\}}
Y_i(g,x).
\]
One can think of $\bar{Y}(g, x)$ as the potential outcome of the mean
phenotype $\bar{Y}$ when we set the genotype as $g$ and environment as
$x$ in an experiment with $n_r$ replications. Furthermore, denote the mean
phenotype for all plants with genotype $g = 1,\dots,n_g$ as
\begin{equation} \label{eq:y-bar-g}
        \bar{Y}(g) = \frac{1}{n_x} \sum_{x=1}^{n_x} \bar{Y}(g, x) =
        \frac{1}{n_r n_x} \sum_{i=1}^n 1_{\{G_i = g\}} Y_i,
\end{equation}
which can be thought as the potential outcome of the mean phenotype $\bar{Y}$ when we set the genotypes as $g$ in $n_r n_x$ replications over $n_x$ environmental conditions.

Let $G$ and $G'$ be two independent random variables that are both
uniformly distributed over $\{1,\dots,n_g\}$, and let $X$ be
uniformly distributed over $\{1,\dots,n_x\}$. Following equation [\ref{eq:def-h2}], the counterfactual heritability of the mean
phenotype $\bar{Y}$ is given by
\[
\xi = \frac{\Var(\bar{Y}(G) - \bar{Y}(G'))}{2 \Var(\bar{Y}(G))}.
\]
We next compare this with an existing definition of plant breeding
heritability assuming $G$ and $X$ are independent by design.

\begin{example}
        Consider the following additive model for the potential outcome of the
        individual phenotype:
        \begin{equation} \label{eq:plant-additive}
                Y(g,x) = \mu + \alpha(g) + \beta(x) + \gamma(g, x) + E,
        \end{equation}
        where $\alpha(g)$ represents the genetic main effect, $\beta(x)$ the
        environmental main effect, $\gamma(g,x)$ the gene-environment
        interaction, and $E$ the idiosyncratic noise. Let
        the variance of each of these terms (under uniformly distributed $G$ and
        $X$) be denoted as
        \[
        \begin{aligned}
        \sigma_G^2 &= \Var(\alpha(G)),~\sigma_X^2 =
        \Var(\beta(X)),
        \sigma_{GX}^2 &= \Var(\gamma(G, X)),~\sigma_E^2 = \Var(E).
        \end{aligned}
        \]
        Using equation [\ref{eq:y-bar-g}], we have
        \[
        \bar{Y}(g) = \mu + \alpha(g) + \frac{1}{n_x} \sum_{x=1}^{n_x}
        \{\beta(x) + \gamma(g, x)\} + \bar{E},
        \]
        where $\Var(\bar{E}) = \sigma_E^2 / (n_x n _r)$. The precise expression of the counterfactual heritability $\xi$ in terms of the variance components depends on our assumptions on the terms in [\ref{eq:plant-additive}]. If we assume a fixed-effects model with the following identification assumptions: $\sum_g \alpha(g) = \sum_x \beta(x) = \sum_g
        \gamma(g,x) = \sum_x \gamma(g,x) = 0$, then
        \[
        \xi = \frac{\sigma_G^2}{\sigma_G^2 + \sigma_E^2 / (n_x n_r)}.
        \]
        In contrast, if we assume a random-effects model in which $\alpha(g), \beta(x), \gamma(g,x)$ are independent, mean-zero random variables, then
        \[
        \xi = \frac{\sigma_G^2+\sigma^2_{GX}/n_x}{\sigma_G^2 + \sigma_X^2 /
                n_x + \sigma^2_{GX}
                / n_x + \sigma_E^2 / (n_x n_r)}.
        \]
        Both of these expressions are slightly different from the
        typical definition of
        plant heritability \cite[eq.\ 6]{schmidt2019heritability}:
        \[
        H_{\text{plant}}^2 = \frac{\sigma_G^2}{\sigma_G^2 + \sigma^2_{GX}
                / n_x + \sigma_E^2 / (n_x n_r)}.
        \]
        Interestingly, \citet{schmidt2019heritability} assumes that the terms
        in [\ref{eq:plant-additive}] are random effects and cited prior work
        for the above expression of $H_{\text{plant}}^2$. A closer
        investigation of the derivation in \citet[p.\ 27-28]{holland2002} shows
        that $H_{\text{plant}}^2$ is derived, however, under a
        fixed-effects model by computing the expectation of sample estimator
        of certain variance components.
\end{example}

The above example shows that existing definitions of heritability in
breeding experiments based on variance components may not have a clear
causal interpretation even under relatively simple models. Adopting
a counterfactual perspective can help researchers state their
assumptions (e.g.\ random effects versus fixed effects) more explictly
and interpret their estimand more easily.
\section{Discussion}
This article provides a unifying, counterfactual perspective on
heritability. As a major advantage, our notion of counterfactual
heritability does not depend on specific statistical modeling
assumptions. Our results and comparisons with existing
notions of heritability highlight the importance
of clarifying the causal structural assumptions and counterfactual
comparisons in reasoning with heritability. This is in line with in
the ongoing ``causal inference revolution'' in statistics and other
disciplines, where the focus of scientific research is shifted from
statistical models and estimators to study designs and estimands/parameters.

A key requirement in defining counterfactual heritability is to find
background environmental or genetic factors $X$ such that the genotype
$G$ is \emph{randomized} conditional on $X$. The appropriate choice of
$X$ depends on the design of the study. In plan breeding experiments,
the natural choice of $X$ is different and may include control
variables used in the design. In observational human genetics study,
the most natural choice of $X$ is the parental haplotypes, leading to
the notion of one-generation counterfactual heritability discussed in
Section \ref{sec:one-generation}. Reassuringly, in twin studies this
basically reduces to familiar Falconer's formula under additive
models. As discussed earlier, the notion of one-generation
heritability is closely related to recent discussion in human genetics
that highlight the distinction between (direct) genetic effects and
(indirect) nurturing effects \citep{kong2018nature,
young2018relatedness}. Our framework can be broadly viewed as a
nonparametric, counterfactual extension to those efforts.

Counterfactual heritability can be further extended to a
notion of counterfactual explainability of other factors. For example,
the same definition can be used to define the counterfactual
explainability of specific genetic factors and can be further extended
to define the explainability of their interactions; similarly, it can
be used to define the heritability of environmental factors and
extended to define the explainability of gene-environment
interaction. This is explored by
\citet{gaoCounterfactualExplainabilityBlackbox2024} beyond the
genetic heritability problem and is closely related to the functional analysis of
variance \citep{hoeffding1948,hookerGeneralizedFunctionalANOVA2007}.

Our notion of counterfactual heritability is naturally defined for
continuous (``quantitative'') and binary (``qualitative'') traits. For
binary traits, it is also possible to define its counterfactual
heritability through a latent, continuous liability index.  
Moreover, one can naturally generalize [\ref{eq:def-h2}] to define  counterfactual the genetic correlation between two phenotypes $Y$ and $Z$ as \[
\rho = \frac{\Cov(Y(G) - Y(G'), Z(G)-  Z(G'))}{2 \sqrt{\Var(Y(G)) \Var(Z(G))}}.
\]
We leave the investigation of these quantities (e.g., their bounds and interpretations) for future work.

When discussing heritability, it is important to bear in mind that
nature vs.\ nurture is a complex question, and
heritability is just one number and provides inherently limited
information. The fact that counterfactual heritability is only partially identified
(can only be bounded using empirical data) reflects our inherent
uncertainty about the nature vs.\ nurture decomposition, due to the
impossibility of ascertaining the interaction strength between genes
and unobserved environmental factors. We believe any notion of
heritability should be interpreted against this backdrop.

\section*{Acknowledgment}
Qingyuan Zhao was supported in part by the Engineering and Physical Sciences
  Research Council (EP/V049968/1). Hongyuan Cao is supported in part by NSF DMS 2311249. The authors thank Yunpeng Wang and
  Matthew Stephens for their comments on oral presentations of this
  work.

\newpage
\appendix
\section{Proofs} \label{sec:proof}
\subsection{Proofs with counterfactual heritability within population}
\numberwithin{equation}{section}
\setcounter{equation}{0}

\numberwithin{table}{section}
\setcounter{table}{0}

\begin{lemma}\label{lemma_basic_eq}
        Let $A, B$ be two random variables with
        distribution functions $F_A$ and $F_B$. Then
        \[
        \E(F^{-1}_A(U) - F'^{-1}_B(U))^2\leq \E(A - B)^2 \leq \E(F_A^{-1}(U) -
        F_B^{-1}(1 - U))^2.
        \]
        where $U$ is a uniform random variable on $[0,1]$ and these inequalities are tight.
\end{lemma}

\begin{proof}
        Note that
        $$
        \begin{aligned}
                \E(A - B)^2&=\E(A^2)+\E(B^2)-2\E(AB)\\&=\E(A^2)+\E(B^2)-2\E(A)\E(B)-2\Cov(A,B).
        \end{aligned}
        $$
        The first three quantities on the right hand side depend on $F_A$ and $F_B$ only, so it suffices to bound $\Cov(A,B)$.

        By the Hoeffding's Covariance Identity \citep{hoeffding1940masstabinvariante}, we have
        $$
        \Cov(A,B)=\iint \{F_{AB}(a,b)-F_A(a)F_B(b)\}dadb
        $$
        where $F_{AB}(a,b)=P(A\le a,B\le b)$ is the distribution function of the random vector $(A,B)$. (This identity can be proved by writing $2\Cov(A,B)=\E\{(A-A')(B-B')\}$ where $(A',B')$ is an IID copy of $(A,B)$ and
        $$A - A' = \int_{A'}^{A}1 \, da = \int_{-\infty}^\infty [1_{\{a\le A\}}-1_{\{a\le A'\}}] \, da$$
        and the same for $B - B'$.)

        The Fréchet–Hoeffding copula bounds \citep{frechet1951tableaux} say the following:
        $$
        \max\{F_A(a)+F_B(b)-1,0\}\le F_{AB}(a,b)\le \min\{F_A(a),F_B(b)\}.
        $$
        The lower bound is tight when $A=F^{-1}_A(U)$ and
        $B=F^{-1}_B(1-U)$, and the upper bound is tight when
        $A=F^{-1}_A(U)$ and $B=F^{-1}_B(U)$, where $U$ is a uniform
        random variable on $[0,1]$. By plugging in these bounds to the
        integral $\int \int F_{AB}(a,b) \, da \, db$, we have the
        upper and lower bounds for $\E(A-B)^2$ stated in the Lemma.
\end{proof}

\begin{proof}[Proof of Proposition \ref{thm:binary-bound}]
        Using [\ref{orth}], we have
        \[
        \E\{Y(G) - Y(G') \mid X\} = 0.
        \]
        Thus, by the law of total variance, we have
        \begin{align*}
                V =& \Var(Y(G)-Y(G')) \\
                =& \E\{\Var(Y(G)-Y(G') \mid X)\} \\
                =&\E\big(\E\big[\{Y(G)-Y(G')\}^2 \mid X\big]\big)\\
                =&2 p_0p_1\E\big(\E\big[\{Y(0)-Y(1)\}^2\mid X\big]\big),
        \end{align*}
        where $p_g = \P(G = g), g=0,1$.

        By Lemma \ref{lemma_basic_eq}, we have
        $$
        \begin{aligned}
                V \geq & 2p_0p_1\E\big(\E\big[\{F^{-1}_{0,X}(U)-F^{-1}_{1,X}(U)\}^2\mid X\big]\big) \\
                = & \E\big(\E \big[\{F^{-1}_{G,X}(U)-F^{-1}_{G',X}(U)\}^2\mid X \big]\big) \\
                = &\E\big[\{F^{-1}_{G,X}(U)-F^{-1}_{G',X}(U)\}^2\big].
        \end{aligned}
        $$
        Similarly, we have
        $$
        V \leq \E\big[\{F^{-1}_{G,X}(U)-F^{-1}_{G,X'}(1-U)\}^2\big].
        $$
        Therefore, we have $\xi_l\le\xi\le \xi_u$ and these inequalities are tight.
\end{proof}

\begin{proof}[Proof of Theorem \ref{thm:main-bound}]
The same bound $\xi_l \leq \xi \leq \xi_u$ can be derived using the law of total variance
$$
\Var(Y)=\E(\Var(Y\mid X))+\Var(\E(Y\mid X)),
$$
and
$$
\begin{aligned}
\E(\Var(Y\mid X))&=\E[\Var\{\E(Y\mid X,G)\mid X\}]+\E[\E\{\Var(Y\mid X,G)\mid X\}]
\\&=\E[\Var\{\E(Y\mid X,G)\mid X\}]+\E\{\Var(Y\mid X,G)\}
\end{aligned}.
$$
By following the same steps as in the proof of Proposition \ref{thm:binary-bound}, it is straightforward to show that
$$
\begin{aligned}
V &= \Var(Y(G)-Y(G'))\\&= \E\big(\E\big[\{Y(G)-Y(G')\}^2 \mid X\big]\big) \\&\ge \E[\{F^{-1}_{G,X}(U)-F^{-1}_{G,X'}(U)\}^2]
\end{aligned}$$
also holds for non-binary $G$. The equality holds when
$Y(G)=F^{-1}_{G,X}(U)$, which is equivalent to the condition that all
potential outcomes are comonotonic given $X$. This establishes the
tight inequality $\xi_l \leq \xi$. Similarly, we have $\xi \leq
\xi_u$, but this upper bound is not tight in general when $G$ takes
on more than two values (when there are three variables, it's not
possible for all three pairs to be countermonotonic).

For the other bounds, consider the alternative expression
\begin{align}
        &\frac{1}{2}\Var(Y(G) - Y(G')) \nonumber \\
        =& \frac{1}{2}\E\{\Var(Y(G) - Y(G') \mid X)\} \nonumber \\
        =& \E\{\Var(Y \mid X)\} - \E\{\Cov(Y(G), Y(G') \mid X)\}. \label{eq:xi-decomposition}
\end{align}
It suffices to bound the second term on the right hand side. The inequality $\xi_{l}' \le \xi$ follows from
\begin{align*}
        &\Cov(Y(G), Y(G') \mid X) \\
        =& \Cov\{ \E(Y(G) \mid G, G', X), \E(Y(G') \mid G, G', X) \mid X \} \\
        &+ \E\{\Cov(Y(G), Y(G') \mid G, G', X) \mid X \} \\
        =& \E\{\Cov(Y(G), Y(G') \mid G, G', X) \mid X \} \\
        \leq& \E \{ \Var(Y(G) \mid G, G', X) \mid X\} \\
        =& \E \{\Var(Y \mid G, X) \mid X \},
\end{align*}
where the first equality follows from the law of total covariance, the
second equality follows from [\ref{orth}] (note that $G$ and $G'$ are
conditionally independent given $X$), the inequality follows from
$\Cov(A,B) \leq \{\Var(A) + \Var(B)\}/2$ for any random variables $A$
and $B$ and symmetry, and the last equality follows from [\ref{orth}]
and consistency of potential outcomes. Therefore, we have
\begin{align*}
    &\frac{1}{2}\Var(Y(G) - Y(G'))\\
    =& \E\{\Var(Y \mid X)\} - \E\{\Cov(Y(G), Y(G') \mid X)\}\\
    \ge & \E\left[\Var(Y \mid X) -  \E \{\Var(Y \mid G, X) \mid X \}\right]\\
    = & \E\left[\Var\{\E(Y\mid G,X)\mid X\}\right].
\end{align*}

The upper bound $\xi \leq \xi_u'$ can be derived similarly by conditioning on the potential outcomes schedule $Y(\cdot) = \{Y(g): g \in \mathcal{G}\}$:
\begin{align*}
        &\Cov(Y(G), Y(G') \mid X) \\
        =& \Cov\{ \E(Y(G) \mid Y(\cdot), X), \E(Y(G') \mid Y(\cdot), X) \mid X \} \\
        &+ \E\{\Cov(Y(G), Y(G') \mid Y(\cdot), X) \mid X \} \\
        =& \Cov\{ \E(Y(G) \mid Y(\cdot), X), \E(Y(G') \mid Y(\cdot), X) \mid X \} \\
        =&\Var\{ \E(Y(G) \mid Y(\cdot), X)\mid X \}\\\ge &0.
\end{align*}
So $V / 2 \leq \E\{\Var(Y \mid X)\}$ and $\xi \leq \xi_u'$.

It remains to prove $\xi_l' \leq \xi_l$.
This follows from
$$
\begin{aligned}
        &\E\left[\{F^{-1}_{G,X}(U)-F^{-1}_{G',X}(U)\}^2\right]
        \\=&\E\left(\E\left[\{F^{-1}_{G,X}(U)-F^{-1}_{G',X}(U)\}^2\mid G,G',X\right]\right)
        \\\ge &\E\big[\big\{\E(F^{-1}_{G,X}(U)\mid G,G',X)-\E(F^{-1}_{G',X}(U)\mid G,G',X)\big\}^2\big]
        \\
        =&\E\big[\big\{\E(Y(G)\mid G,G',X)-\E(Y(G')\mid G,G',X)\big\}^2\big]\\
        =& 2\E\left\{[\E(Y(G)\mid G,G',X)]^2-\E(Y(G)\mid G,X)\E(Y(G')\mid G',X)\right\}\\
        =& 2\E\left(\E\left\{[\E(Y(G)\mid G,X)]^2\mid X\right\}-\big[\E\{\E(Y(G)\mid G,X)\mid X\}\big]^2\right)\\=& 2\E\left[\Var\{\E(Y\mid G,X)\mid X\}\right].
\end{aligned}
$$
The equality $\xi_l' = \xi_l$ holds when the quantile difference $F^{-1}_{G,X}(u)-F^{-1}_{G',X}(u)$ is a constant (does not depend on $u$ but may depend on $G, G', X$).
\end{proof}

\begin{lemma} \label{lem:xil-h2}
If $X\indep G$, then $\xi_l'\ge H^2$.
\end{lemma}

\begin{proof}
Consider Hoeffding's decomposition \citep{hoeffding1948}:
$$Y=\mu+f_1(X)+f_2(G)+R,$$
where $\mu=\E(Y)$, $f_1(X)=\E(Y\mid X)-\mu$, $f_2(G)=\E(Y\mid G)-\mu$, and $R=Y-\mu-f_1(X)-f_2(G)$. Denote $h(X,G)=\E(R\mid X,G)$. So $\E(f_1(X)) = \E(f_2(G)) = 0$ and $\E(h(X,G) | X) = \E(h(X,G) \mid G) = 0$.

Thererfore,
$$
\begin{aligned}
        &\E\{\Var(\E(Y\mid X,G)\mid X)\}\\=&\E\{\Var(f_2(G)+h(X,G)\mid X)\}
        \\=&\E\{\Var(f_2(G)\mid X)+\Var(h(X,G)\mid X)\}\\&+2\E\{\Cov(f_2(G),h(X,G)\mid X)\}
        \\=&\E\{\Var(f_2(G)\mid X)\}+\Var(h(X,G)\mid X)
        \\\ge& \Var(f_2(G)) \\
        =&\Var(\E(Y \mid G)).
\end{aligned}
$$
The third equality above is true because, by using $X \indep G$,
$$
\begin{aligned}
        &\E\{\Cov(f_2(G),h(X,G)\mid X)\}
        \\=&\E\{\E(f_2(G) h(X,G)\mid X)\}-\E\{\E(f_2(G) \mid X)\E(h(X,G) \mid X)\}
        \\=&\E\{\E(f_2(G) h(X,G)\mid X)\}
        \\=& \E(f_2(G)h(X,G))
        \\=&\E\{\E(f_2(G)h(X,G)\mid G)\}
        \\=&\E \{f_2(G)\cdot \E(h(X,G)\mid G)\}
        \\=&0.
\end{aligned}
$$
Thus, we have $\xi_l' \geq H^2$.
\end{proof}

\begin{lemma}
\label{lemma:decomp}
Under [\ref{orth}], we have
    \begin{equation}
        \label{eq:decomp}
        1-\xi=\Cor(Y(G),Y(G'))=\frac{\Var[\E\{Y\mid X\}]+\E[\Var\{E(Y\mid Y(\cdot), X) \mid X\}]}{\Var (Y)}.
    \end{equation}
Specially, if $X$ is empty (so [C1] is true),
     $$
    1-\xi=\frac{\Var\{\E(Y(G)\mid Y(\cdot))\}}{\Var (Y)}.
    $$
\end{lemma}

\begin{proof}
    Note that under [\ref{orth}], $G$ and $G'$ marginally have the same distribution, and
$$
\Var(Y(G)-Y(G'))=2\Var(Y(G))-2\Cov(Y(G),Y(G')).
$$ Therefore, we have
    $$
        \begin{aligned}
                \xi&=\frac{\Var(Y(G)-Y(G'))}{2\Var(Y)}
                \\=&\frac{2\Var(Y(G))-2\Cov(Y(G),Y(G'))}{2\Var(Y(G))}
                \\=&1-\Cor(Y(G),Y(G')).
        \end{aligned}
        $$
    The first equality holds.

   By the law of total covariance, 
$\Cov(Y(G),Y(G'))$ can be decomposed into
$$
\begin{aligned}
     &\Cov(Y(G),Y(G'))
    \\=& \Cov\left[\E\{Y(G)\mid X\},\E\{Y(G')\mid X\}\right]+\E\{\Cov(Y(G),Y(G')\mid X)\} \\
    =& \Var\left[\E\{Y(G)\mid X\}\right] + \E\left[\Var\{\E(Y(G) \mid Y(\cdot), X) \mid X\}\right]\\
    &+\E(\E\left[\E\{\Cov(Y(G),Y(G')\mid X, Y(\cdot) \} \mid X\right]) \\
=& \Var\left[\E\{Y(G)\mid X\}\right] + \E\left[\Var\left\{\E(Y(G) \mid Y(\cdot), X) \mid X\right\}\right].
\end{aligned}
$$
Therefore, we obtain \eqref{eq:decomp}.
\end{proof}

\subsection{Derivation of the result in Table \ref{tab:heritability1}}
\begin{proof}

  For $Y(g)=\beta_1g_1+\beta_2g_2+E_1$ the $\xi_u'=1$, and
  \[
    h^2=H^2=\xi=\xi_l=\xi_{l}'=\frac{\Var(\beta_1g_1+\beta_2g_2)}{\Var(Y)}=\frac{\beta_1^2+\beta_2^2}{\beta_1^2+\beta_2^2+1}
  .\]

        For $Y(g)=\beta g_1g_2+E_1$, we have $$\Var(\beta G_1G_2)=\beta^2 \{\E(G_1^2G_2^2)-\E^2(G_1G_2)\}=\frac{3}{16}\beta^2,$$ and
        $$\Var(Y)=\Var(\beta G_1G_2+E_1)=\Var(\beta G_1G_2)+\Var(E_1)=\frac{3\beta^2}{16}+\frac{1}{4}.$$
        In this case, $\xi_u'=1$ and   $$H^2=\xi=\xi'_l=\xi_l=\frac{\frac{3}{16}\beta^2}{\frac{3}{16}\beta^2+\frac{1}{4}}=\frac{3\beta^2}{3\beta^2+4}.$$
        Let $$\tilde \beta=\operatorname{argmin}_{\beta_0,\beta_1,\beta_2} \E\{(\beta G_1G_2+E_1-\beta_0-\beta_1G_1-\beta_2G_2)^2\}.$$
        By taking the derivative, we have
        $$
        \E\left\{\begin{pmatrix}1&G_1&G_2\\G_1&G_1^2&G_1G_2\\G_2&G_1G_2&G_2^2\end{pmatrix}\begin{pmatrix}\tilde \beta_0\\\tilde \beta_1\\\tilde\beta_2\end{pmatrix}\right\}=\beta\cdot \E\left\{\begin{pmatrix}G_1G_2\\G_1^2G_2\\G_1G_2^2\end{pmatrix}\right\}.
        $$
        Therefore,
        $$
        (\tilde \beta_0,\tilde\beta_1,\tilde\beta_2)^T=\begin{pmatrix}1&\frac{1}{2}&\frac{1}{2}\\\frac{1}{2}&\frac{1}{2}&\frac{1}{4}\\\frac{1}{2}&\frac{1}{4}&\frac{1}{2}\end{pmatrix}^{-1}\begin{pmatrix}\frac{1}{4}\\\frac{1}{4}\\\frac{1}{4}\end{pmatrix}\cdot\beta=(-\frac{\beta}{2},\frac{\beta}{2},\frac{\beta}{2})^T.
        $$
        We obtain that $$h^2=\frac{\Var(\tilde \beta_1g_1+\tilde\beta_2g_2)}{\Var(Y)}=\frac{2\cdot\frac{\beta^2}{4}\frac{1}{4}}{\frac{3\beta^2}{16}+\frac{1}{4}}=\frac{2\beta^2}{4+3\beta^2}.$$

        For $Y(g)=\beta g_1E_2+E_1$. We find that  $\E(Y\mid G)=0$, so that $H^2=0$, $\xi_l'=0$ and $h^2=0$. The variance of $Y$ is
        $$
        \begin{aligned}\Var(Y)&=\Var(\beta g_1E_2)+\Var(E_1)\\&=\beta^2\{\E(G_1^2E_2^2)-\E^2(G_1 E_2)\}+\frac{1}{4}\\&=\frac{\beta^2}{8}+\frac{1}{4}.\end{aligned}$$
        We find that
        $$
        \begin{aligned}
                \Var(Y(G)-Y(G'))&=\Var(\beta(G-G')E_2)
                \\&=\beta^2[\E\{(G-G')^2E_2^2\}-\E^2\{(G-G')E_2\}]
                \\&=\frac{1}{4}\beta^2 \E\{(G-G')^2\}
                \\&=\frac{1}{4}\beta^2 \Var(G-G')
                \\&=\frac{1}{2}\beta^2 \Var(G)
                \\&=\frac{\beta^2}{8}.
        \end{aligned}
        $$
        Therefore $\xi=\frac{\Var(Y(G)-Y(G'))}{2\Var(Y)}=\frac{\frac{\beta^2}{8}}{2(\frac{\beta^2}{8}+\frac{1}{4})}=\frac{\beta^2}{2\beta^2+4}$.

        Note that $Y\mid G\sim N(0,\frac{1}{4}\{1+\beta^2G^2\})$, so that $F^{-1}_{G}(U)=\frac{1}{2}\sqrt{1+\beta^2G^2}\Phi^{-1}(U)$, where $\Phi^{-1}$ is the quantile function of standard normal distribution. Therefore we have
        $$
        \begin{aligned}
                &\E\left[\{(F^{-1}_{G}(U)-F^{-1}_{G'}(U)\}^2\right]
                \\=&2\cdot\frac{1}{4}\E\left[\{(F^{-1}_{1}(U)-F^{-1}_{0}(U)\}^2\right]
                \\=&\frac{1}{2}\E\left[\{\frac{1}{2}(\sqrt{1+\beta^2}-1)\Phi^{-1}(U)\}^2\right]
                \\=&\frac{1}{8}(\sqrt{1+\beta^2}-1)^2 \E[\{\Phi^{-1}(U)\}^2]
                \\=&\frac{1}{8}(\sqrt{1+\beta^2}-1)^2
        \end{aligned}
        $$
        Therefore,
        $$
        \begin{aligned}
                \xi_l&=\frac{\E\left[\{(F^{-1}_{G}(U)-F^{-1}_{G'}(U)\}^2\right]}{2\Var(Y)}\\&=\frac{\frac{1}{8}(\sqrt{1+\beta^2}-1)^2}{2(\frac{\beta^2}{8}+\frac{1}{4})}\\&=\frac{(\sqrt{1+\beta^2}-1)^2}{2\beta^2+4}.
        \end{aligned}
        $$
        The upper bound $\xi_u'=1.$

        For $Y(g)=\beta_1g_1+\beta_2X+E_1$, we have
        $$h^2=H^2=\xi=\xi_l'=\xi_l=\frac{\Var(\beta_1g_1)}{\Var(Y)}=\frac{\beta_1^2}{\beta_1^2+\beta_2^2+1}.$$ The upper bound $\xi_u'$ is derived by
        $$
        \xi_u'=1-\frac{\Var(\E(Y\mid X))}{\Var(Y)}= 1-\frac{\Var(X\beta_2)}{\Var(Y)}=\frac{\beta_1^2+1}{\beta_1^2+\beta_2^2+1}.
        $$

        Now we concentrate on the case when $Y(g)=\beta g_1X+E_1$. With the same approach as the case when $Y(g)=\beta g_1E_2+E_1$, we can show that $h^2=0$, $H^2=0$ and $\xi=\frac{\beta^2}{4+2\beta^2}$. Besides, we have $\xi=\xi_l'=\xi_l=\frac{\beta^2}{4+2\beta^2}$ by comonotonicity.
        As $\Var(\E(Y\mid X))=\Var(\frac{1}{2}\beta X)=\frac{\beta^2}{16}$,
        the upper bound $\xi_u'$ can be calculated by
        $$
        \xi_u'=1-\frac{\Var(\E(Y\mid X))}{\Var(Y)}=1-\frac{\frac{1}{16}\beta^2}{\frac{\beta^2}{8}+\frac{1}{4}}=\frac{4+\beta^2}{4+2\beta^2}.
        $$
\end{proof}

\subsection{Relationship between twin heritability and the counterfactual heritability}

We assume that $g_m,g_f\in \{0,1,2\}$ are the mother and father's genes which are generated from the population. We assume that $Y_1(g)=f(g)+\epsilon_1$, and $Y_2(g)=f(g)+\epsilon_2$, where
$$
f(g)=\begin{cases}-a&g=0\\d&g=1\\a&g=2\end{cases}
$$
and $\epsilon_1,\epsilon_2$ are i.i.d. The children's genes are generated based on the parents' genes. The detail can be found in Table \ref{table:gene_2}. Denote $G_1$ and $G_2$ as the genotypes of participants. For MZ twin $G_1=G_2=G$ a.s., and for DZ twin $G_1\perp G_2\mid g_m,g_f$. It can be shown that
$$
\Var(f(G))=\sigma_a^2+\sigma_d^2,
$$
and
\begin{equation}
\label{eq:full_sib}
\Cov(f(G_1),f(G_2))=\frac{1}{2}\sigma_a^2+\frac{1}{4}\sigma_d^2.
\end{equation}

\begin{table*}[t]
        \centering
        \caption{The relationship between parents' geneotype and children's genotype and phenotype}
        \label{table:gene_2}
        \begin{tabular}{cccccc}
                \toprule
                Parent & Frequency & Child AA & Child Aa & Child aa & Mean phenotype \\
                \midrule
                AA+AA & $p^4$ & $1$ & $0$ & $0$ & $f({AA})$ \\
                AA+Aa & $4p^3q$ & $\frac{1}{2}$ & $\frac{1}{2}$ & $0$ & $\frac{1}{2} f({AA})+\frac{1}{2}f({Aa})$ \\
                AA+aa & $2p^2q^2$ & $0$ & $1$ & $0$ & $f({Aa})$ \\
                Aa+Aa & $4p^2q^2$ & $\frac{1}{4}$ & $\frac{1}{2}$ & $\frac{1}{4}$ & $\frac{1}{4} f({AA})+\frac{1}{2}f({Aa})+\frac{1}{4} f({aa})$ \\
                Aa+aa & $4pq^3$ & $0$ & $\frac{1}{2}$ & $\frac{1}{2}$ & $\frac{1}{2}f({Aa})+\frac{1}{2} f({aa})$ \\
          aa+aa & $q^4$ & $0$ & $0$ & $1$ & $f({aa})$ \\
          \bottomrule
        \end{tabular}
\end{table*}

The covariance between MZ twins is
$$
\begin{aligned}
        \rho_{MZ}&=\frac{\Cov(Y_1(G),Y_2(G))}{\Var(Y)}\\&=\frac{\Var(f(G))}{\Var(Y)}\\&=\frac{\sigma_a^2+\sigma_d^2}{\Var(Y)},
\end{aligned}
$$
The covariance between DZ twins is
$$
\begin{aligned}
        \rho_{DZ}&=\frac{\Cov(Y_1(G_1),Y_2(G_2))}{\Var(Y)}\\&=\frac{\Cov(f(G_1),f(G_2))}{\Var(Y)}\\&=\frac{\frac{1}{2}\sigma_a^2+\frac{1}{4}\sigma_d^2}{\Var(Y)}.
\end{aligned}
$$

The twin heritability is
$$
h^2_{\text{twin}}=2(\rho_{MZ}-\rho_{DZ})=(\sigma_a^2+\frac{3}{2}\sigma_d^2)/\Var(Y).
$$

The counterfactual heritability is defined as
$$
\begin{aligned}
        \xi&=\frac{\Var(Y_1(G_1)-Y_1(G_2))}{2\Var(Y)}\\&=\frac{\Var(f(G_1)-f(G_2))}{2\Var(Y)}\\&=\frac{\Var\{f(G_1)\}-\Cov\{f(G_1),f(G_2)\}}{\Var(Y)}\\&=\frac{\frac{1}{2}\sigma_a^2+\frac{3}{4}\sigma_d^2}{\Var(Y)}.
\end{aligned}
$$

In fact, the following lemma holds for the relationship between twin heritability $h^2_{\text{twin}}$ and $\xi_l'$.
\begin{lemma}
        \label{lemma:twin}
        Under [\ref{orth}], suppose $Y_1(g)$ and $Y_2(g)$ are potential outcomes of two siblings in a family with $Y_1(G)\mid X\overset{d}{=}Y(G)\mid X$ and $Y_2(G)\mid X\overset{d}{=}Y(G)\mid X$, $Y_1(G)\perp Y_2(G)\mid G,X$, and $Y_1(G)\perp Y_2(G')\mid X$. We have $h^2_{\text{twin}}=2\xi_l'$.
\end{lemma}
\begin{proof}
        For a pair of MZ twins, by the total law of covariance, we have
        $$
        \begin{aligned}
                &\Cov(Y_1(G),Y_2(G))\\=&\E\{\Cov(Y_1(G),Y_2(G)\mid X)\}+\Cov\{\E(Y_1(G)\mid X),\E(Y_2(G)\mid X)\}
                \\=&\E\{\Cov(Y_1(G),Y_2(G)\mid X)\}+\Var\{\E(Y(G)\mid X)\}
                \\=&\E[\Cov\{\E(Y_1(G)\mid G,X),\E(Y_2(G)\mid G,X)\mid X\}]
                \\&+\E[\E\{\Cov(Y_1(G),Y_2(G)\mid G,X)\mid X\}] +\Var\{\E(Y(G)\mid X)\}
                \\=&\E[\Var\{\E(Y(G)\mid G,X)\mid X\}] +\Var\{\E(Y(G)\mid X)\}.
        \end{aligned}
        $$
        For a pair of DZ twins, we have
        $$
        \begin{aligned}
                &\Cov(Y_1(G),Y_2(G'))\\=&\E\{\Cov(Y_1(G),Y_2(G')\mid X)\}+\Cov\{\E(Y_1(G)\mid X),\E(Y_2(G')\mid X)\}
                \\=&\Var\{\E(Y(G)\mid X)\}.
        \end{aligned}
        $$
        Therefore, we find
        $$
        \begin{aligned}
                h^2_{\text{twin}}=&2\frac{\Cov(Y_1(G),Y_2(G))-\Cov(Y_1(G),Y_2(G'))}{\Var(Y)}
                \\=&2\frac{\E[\Var\{\E(Y(G)\mid G,X)\mid X\}]}{\Var(Y)}
                \\=&2\xi_l'.
        \end{aligned}$$
\end{proof}

Here is the derivation of the equation \eqref{eq:full_sib}. Note $G_m$ and $G_f$ are the mother's and father's genotype, $G_1$ and $G_2$ are the genotypes of two siblings. Define $g(G)=f(G)-\E(f(G))$ to be the centralized genetic effect. We have $$p^2g(AA)+2pqg(Aa)+q^2g(aa)=0.$$
Note that $G_1\perp G_2\mid G_m,G_f$ and $G_1\overset{d}{=}G_2\mid G_m,G_f$.
We obtain that
$$
\begin{aligned}
        &\Cov(f(G_1),f(G_2))\\=&\E(g(G_1)g(G_2))
        \\=&\E[\E\{g(G_1)g(G_2)\mid G_m,G_f\}]
    \\=&\E[\E\{g(G_1)\mid G_m,G_f\}\E\{g(G_2)\mid G_m,G_f\}]
        \\=&p^4 \{g(AA)\}^2+4p^3q \{\frac{1}{2} g({AA})+\frac{1}{2}g({Aa})\}^2+2p^2q^2\{g({Aa})\}^2\\&\quad+4p^2q^2\{\frac{1}{4} g({AA})+\frac{1}{2}g({Aa})+\frac{1}{4} g({aa})\}^2\\&\quad+4pq^3\{\frac{1}{2}g({Aa})+\frac{1}{2} g({aa})\}^2+q^4\{g({aa})\}^2
        \\&\quad -\{p^2g(AA)+2pqg(Aa)+q^2g(aa)\}^2
        \\=&p^3q\{g({AA})-g({Aa})\}^2+pq^3\{g({aa})-g({Aa})\}^2\\&\quad+p^2q^2\{\frac{\{g(AA)+2g(Aa)+g(aa)\}^2}{4}-2g(Aa)^2-2g(AA)g(aa)\}
        \\=&\quad p^2q^2\left\{\frac{\{x+y+4g(Aa)\}^2}{4}-2g^2(Aa)-2\{x+g(Aa)\}\{y+g(Aa)\}\right\}\\&\quad +pq(p^2x^2+q^2y^2)
        \\=&pq[p^2x^2+q^2y^2+pq\{\frac{(x+y)^2}{4}-2xy\}]
\end{aligned}
$$
where $x=g(AA)-g(Aa)=a-d$, and $y=g(aa)-g(Aa)=-a-d$. It can be simplified by
$$
\begin{aligned}
        &\Cov(f(G_1),f(G_2))\\=&pq\{p^2(a-d)^2+q^2(-a-d)^2\}\\&\quad+p^2q^2\{\frac{(a-d-a-d)^2}{4}-2(a-d)(-a-d)\}
        \\=&pq\{p^2(a-d)^2+q^2(a+d)^2+pq(2a^2-d^2)\}
        \\=&pq\{(p^2+q^2+2pq)a^2+2(q^2-p^2)ad+(p^2+q^2-pq)d^2\}
        \\=&pq\{a^2+2(q-p)ad+(p^2+q^2-pq)d^2\}
        \\=&pq\{a+(q-p)d\}^2+p^2q^2d^2
        \\=&\sigma_a^2/2+\sigma_d^2/4.
\end{aligned}
$$

\subsection{Simulation details for Table \ref{tab:heritability2} and \ref{tab:heritability3}}
In this section, we provide details for deriving results in Table \ref{tab:heritability2} and \ref{tab:heritability3}. We first introduce some notations.
Consider two independent genes with minor allele frequency (MAF) $0.5$. Denote parents' genotypes as $G_p=(G_{1f},G_{1m},G_{2f},G_{2m})$, where the subscripts $f$ and $m$ indicate ``father'' and ``mother'', respectively. Each element of $G_p$ is i.i.d 
following distribution $\mathcal{P}_G$, with probabilities $0.25$, $0.5$, and $0.25$ for values $0$, $1$, and $2$, respectively. Conditional on parents' genotype, we can obtain children's genotype $G=(G_1,G_2)$ from Table \ref{table:gene_2}.  
Averaging over the population, marginally, $G_1\sim \mathcal P_G$ and $G_2\sim \mathcal P_G.$ 

\sloppy Let $X$ be the covariate, which 
contains
environmental variables in Table \ref{tab:heritability2} or parents' phenotypes in Table \ref{tab:heritability3}. In Table \ref{tab:heritability2}, $X\sim P_X$ 
independently of $G$  as an external environmental effect, where 
$\mathcal P_X(0),\mathcal P_X(0.5),\mathcal P_X(1),\mathcal P_X(1.5)$ and $\mathcal P_X(2)$ are $1/16,4/16,6/16,4/16$, and $1/16,$ respectively. The error terms, $E_1\sim N(\mu=0,\sigma^2=0.5)$ and $E_2\sim P_X$. Moreover, $E_1$ and $E_2$ are independent 
of $G$ and $G_p,$ and $G'=(G'_1,G'_2)$ is an i.i.d copy of $G$. We calculate 
heritability based on different definitions (e.g. narrow, broad, counterfactual) and obtain the lower and upper bounds of the counterfactual heritability. The results are summarized in Table \ref{tab:heritability2}.

In Table \ref{tab:heritability3}, the setting is exactly the same as that of Table \ref{tab:heritability2}, except that the covariate $X$ is defined as $X=(G_{1f}+G_{1m})/2$, representing the average parental genotype (i.e., a proxy for parental phenotype). This 
introduces a correlation between children's genotype and environmental factors. Additionally, $G'$ is an i.i.d copy of $G$ conditional on $G_p$, mimicking heritability within one generation. We also calculate $h^2_{\text{twin}}$ defined in (\ref{def:twin}).  
All results are presented in Table \ref{tab:heritability3}.

In the simulation, we 
know the 
data-generating process, which allows us to compute the true values of heritabilities in all cases except $\xi_l$. When we calculate $\xi$, $\xi_l'$, $\xi_u'$ and $h^2_{\text{twin}}$, which are defined in (\ref{eq:def-h2}),(\ref{eq:xi_lower}),(\ref{eq:xi_upper}) and (\ref{def:twin}), it suffices to evaluate $\E(Y\mid G,G_p)$, $\Var(Y\mid G,G_p)$, $\E(Y(G)-Y(G')\mid G,G',G_p)$ and $\Var(Y(G)-Y(G')\mid G,G',G_p)$. The results of these expressions are provided in subsection \ref{sec:detailed_ev}. The calculations of narrow-sense heritability $h^2,$ defined in (\ref{def:h2}), and broad-sense heritability $H^2$, defined in 
(\ref{def:H2}), involve calculating $\E(Y\mid G)$, which can be obtained from $\E(Y\mid G, G_p)$ and the posterior probability $P(G_p\mid G)$.

We next elaborate on the calculation of $\xi_l$, defined in (\ref{eq:xi-l}). In the binary case, the quantiles $F^{-1}_{G,X}(U)$ and $F^{-1}_{G',X}(U)$ (as defined in Proposition \ref{thm:binary-bound}) can be calculated directly by $P(Y=1\mid G,X)$. In the continuous case, when the tight lower bound $\xi_l$ does not have the same value as $\xi_l'$ (e.g. scenarios $(3)$ of Table \ref{tab:heritability3}), the calculation of $\xi_l$ involves the quantile functions $F^{-1}_{G,X}(U)$ and $F^{-1}_{G',X}(U)$. These are the quantiles of a mixture of normal distributions, and we resort to a numerical approximation. Specifically, for each fixed combination of the parents' genotype ($G_p$), we can generate two sets of potential outcomes: $Y_{1i}(G)\mid G_p$, and $Y_{2i}(G')\mid G_p$, for $i=1,\cdots ,N$, where $N=1,000,000$. We then sort each sample to obtain the order statistics $Y_{1(i)}$, and $Y_{2(i)}$. 
By law of large numbers,
$$\frac{1}{N}\sum\limits_{i=1}^N (Y_{1(i)}-Y_{2(i)})^2\to E\big[\big(F^{-1}_{G_p,G}(U)-F^{-1}_{G_p,G'}(U)\big)^2\mid G_p,G,G'\big].$$ The tight lower bound $\xi_l$ can then be estimated from (\ref{eq:xi-l}).


\subsubsection{Conditional expectations and conditional variances}
\label{sec:detailed_ev}
Here are the results of $\E(Y\mid G,G_p)$, $\Var(Y\mid G,G_p)$, $\E(Y(G)-Y(G')\mid G,G',G_p)$ and $\Var(Y(G)-Y(G')\mid G,G',G_p)$ in different settings of Table \ref{tab:heritability3}. Denote $\Phi(\cdot)$ as the cdf of a standard normal distribution.

When $Y=G_1+G_2+E_1,$
$$
\begin{aligned}
        \E(Y\mid G,G_p)&=G_1+G_2, \\
        \Var(Y\mid G,G_p)&=\Var(E_1)=0.5,\\
        \E(Y(G)-Y(G')\mid G,G',G_p)&=(G_1+G_2)-(G_1'+G_2'),\\
        \Var(Y(G)-Y(G')\mid G,G',G_p)&=0.
\end{aligned}
$$

When $Y=G_1G_2+E_1,$
$$
\begin{aligned}
        \E(Y\mid G,G_p)&=G_1G_2, \\
        \Var(Y\mid G,G_p)&=\Var(E_1)=0.5,\\
        \E(Y(G)-Y(G')\mid G,G',G_p)&=G_1G_2-G_1'G_2',\\
        \Var(Y(G)-Y(G')\mid G,G',G_p)&=0.
\end{aligned}
$$

When $Y=G_1E_2+E_1,$
$$
\begin{aligned}
        \E(Y\mid G,G_p)&=G_1\E(E_2)=G_1,\\
        \Var(Y\mid G,G_p)&=\Var(E_1)+G_1^2\Var(E_2)=0.5+0.25G_1^2,\\
        \E(Y(G)-Y(G')\mid G,G',G_p)&=(G_1-G_1')\E(E_2) = G_1 - G_1',\\
        \Var(Y(G)-Y(G')\mid G,G',G_p)&=(G_1-G_1')^2\Var(E_2) = 0.25(G_1 - G_1')^2.
\end{aligned}
$$

When $Y=G_1+X+E_1,$
$$
\begin{aligned}
        \E(Y\mid G,G_p)&=G_1+X=G_1+\frac{G_{1f}+G_{2f}}{2},\\
        \Var(Y\mid G,G_p)&=\Var(E_1)=0.5,\\
        \E(Y(G)-Y(G')\mid G,G',G_p)&=G_1-G_1',\\
        \Var(Y(G)-Y(G')\mid G,G',G_p)&=0.
\end{aligned}
$$

When $Y=G_1X+E_1,$
$$
\begin{aligned}
        \E(Y\mid G,G_p)&=G_1X=G_1\frac{G_{1f}+G_{1m}}{2},\\
        \Var(Y\mid G,G_p)&=\Var(E_1)=0.5,\\
        \E(Y(G)-Y(G')\mid G,G',G_p)&=(G_1-G_1')X=(G_1-G_1')\frac{G_{1f}+G_{1m}}{2},\\
        \Var(Y(G)-Y(G')\mid G,G',G_p)&=0.
\end{aligned}
$$

When $Y=1\{G_1+G_2+E_1>0\}$,
$$
\begin{aligned}
        \E(Y\mid G,G_p)&=\Phi(\sqrt2 (G_1+G_2)),\\
        \Var(Y\mid G,G_p)&=\E(Y\mid G,G_p)(1-\E(Y\mid G,G_p)),\\
        \E(Y(G)-Y(G')\mid G,G',G_p)&=\Phi(\sqrt2 (G_1+G_2))-\Phi(\sqrt2 (G_1'+G_2')),\\
        \Var(Y(G)-Y(G')\mid G,G',G_p)&=\mbox{abs}\{\E(Y(G)-Y(G')\mid G,G',G_p)\}\\&\quad\cdot(1-\mbox{abs}\{\E(Y(G)-Y(G')\mid G,G',G_p)\}).
\end{aligned}
$$

When $Y=1\{G_1G_2+E_1>0\}$,
$$
\begin{aligned}
        \E(Y\mid G,G_p)&=\Phi(\sqrt2 (G_1G_2)),\\
        \Var(Y\mid G,G_p)&=\E(Y\mid G,G_p)(1-\E(Y\mid G,G_p)),\\
        \E(Y(G)-Y(G')\mid G,G',G_p)&=\Phi(\sqrt2 (G_1G_2))-\Phi(\sqrt2 (G_1'G_2')),\\
        \Var(Y(G)-Y(G')\mid G,G',G_p)&=\mbox{abs}\{\E(Y(G)-Y(G')\mid G,G',G_p)\}\\&\quad\cdot(1-\mbox{abs}\{\E(Y(G)-Y(G')\mid G,G',G_p)\}).
\end{aligned}
$$

When $Y=1\{G_1E_2+E_1>0\}$,
$$
\begin{aligned}
        \E(Y\mid G,G_p)&=\E_{E_2}\{\Phi(\sqrt2 G_1E_2)\}
    \\&=\frac{1}{16}\Phi(0)+\frac{4}{16}\Phi(\frac{\sqrt2}{2} G_1)+\frac{6}{16}\Phi(G_1)\\&\quad+\frac{4}{16}\Phi(\frac{3\sqrt2}{2} G_1)+\frac{1}{16}\Phi(2 G_1), \\
        \Var(Y\mid G,G_p)&=\E(Y\mid G,G_p)(1-\E(Y\mid G,G_p)),\\
        \E(Y(G)-Y(G')\mid G,G',G_p)&=\E_{E_2}\{\Phi(\sqrt2 G_1E_2)\}-\E_{E_2}\{\Phi(\sqrt2 G_1'E_2)\}\\
    &=\frac{4}{16}T(1/2)+\frac{6}{16}T(1)+\frac{4}{16}T(3/2)+\frac{1}{16}T(2),\\
        \Var(Y(G)-Y(G')\mid G,G',G_p)&=*,
\end{aligned}
$$
where
$$
\begin{aligned}
        *=&\Var(Y(G)-Y(G')\mid G,G',G_p)\\=& \Var\{\E(Y(G)-Y(G')\mid G,G',G_p,E_2)\mid G,G',G_p\} \\&\quad+ \E\{\Var(Y(G)-Y(G')\mid G,G',G_p,E_2)\mid G,G',G_p\}
        \\=& \Var_{E_2}(T\mid G,G',G_p) + \E_{E_2}\{|T|(1-|T|)\mid G,G',G_p\},
\end{aligned}
$$
and $T(E_2)=\Phi(\sqrt{2}G_1E_2)-\Phi(\sqrt{2}G_1'E_2).$ Here, the direct calculation shows that
$$
\begin{aligned}
    &\Var_{E_2}(T\mid G,G',G_p)
    \\&=\frac{4}{16}T^2(1/2)+\frac{6}{16}T^2(1)+\frac{4}{16}T^2(3/2)+\frac{1}{16}T^2(2)\\&\quad-\{\E(Y(G)-Y(G')\mid G,G',G_p)\}^2,
\end{aligned}
$$
and
$$
\begin{aligned}
    &\E_{E_2}\{|T|(1-|T|)\mid G,G',G_p\}
    \\&=\frac{4}{16}|T(1/2)|\{1-|T(1/2)|\}+\frac{6}{16}|T(1)|\{1-|T(1)|\}\\&\quad+\frac{4}{16}|T(3/2)|\{1-|T(3/2)|\}+\frac{1}{16}|T(2)|\{1-|T(2)|\}.
\end{aligned}
$$

When $Y=1\{G_1+X+E_1>0\}$,
$$
\begin{aligned}
        \E(Y\mid G,G_p)&=\Phi(\sqrt2 (G_1+X))=\Phi(\sqrt2 (G_1+\frac{G_{1f}+G_{1m}}{2})),\\
        \Var(Y\mid G,G_p)&=\E(Y\mid G,G_p)(1-\E(Y\mid G,G_p)),\\
        \E(Y(G)-Y(G')\mid G,G',G_p)&=\Phi(\sqrt2 (G_1+X))-\Phi(\sqrt2 (G_1'+X))\\
    &=\Phi\{\sqrt{2}(G_1+\frac{G_{1f}+G_{1m}}{2})\}-\Phi\{\sqrt{2}(G_1'+\frac{G_{1f}+G_{1m}}{2})\},\\
        \Var(Y(G)-Y(G')\mid G,G',G_p)&=\mbox{abs}\{\E(Y(G)-Y(G')\mid G,G',G_p)\}\\&\quad\cdot(1-\mbox{abs}\{\E(Y(G)-Y(G')\mid G,G',G_p)\}).
\end{aligned}
$$

When $Y=1\{G_1X+E_1>0\}$,
$$
\begin{aligned}
        \E(Y\mid G,G_p)&=\Phi(\sqrt2 G_1X)=\Phi(\sqrt2 G_1\frac{G_{1f}+G_{1m}}{2}),\\
        \Var(Y\mid G,G_p)&=\E(Y\mid G,G_p)(1-\E(Y\mid G,G_p)),\\
        \E(Y(G)-Y(G')\mid G,G',G_p)&=\Phi(\sqrt2 G_1X)-\Phi(\sqrt2 G_1'X)\\
    &=\Phi(\sqrt{2}G_1\frac{G_{1f}+G_{1m}}{2})-\Phi(\sqrt{2}G_1'\frac{G_{1f}+G_{1m}}{2}),\\
        \Var(Y(G)-Y(G')\mid G,G',G_p)&=\mbox{abs}\{\E(Y(G)-Y(G')\mid G,G',G_p)\}\\&\quad\cdot(1-\mbox{abs}\{\E(Y(G)-Y(G')\mid G,G',G_p)\}).
\end{aligned}
$$


\begin{thebibliography}{}

\bibitem[\protect\citeauthoryear{Bareinboim, Correa, Ibeling, and
  Icard}{Bareinboim et~al.}{2022}]{bareinboimPearlHierarchyFoundations2022}
Bareinboim, E., J.~D. Correa, D.~Ibeling, and T.~Icard (2022, March).
\newblock On {{Pearl}}'s {{Hierarchy}} and the {{Foundations}} of {{Causal
  Inference}}.
\newblock In {\em Probabilistic and {{Causal Inference}}: {{The Works}} of
  {{Judea Pearl}}\/} (1 ed.), Volume~36, pp.\  507--556. New York, NY, USA:
  Association for Computing Machinery.

\bibitem[\protect\citeauthoryear{Barry, Walker, Cheesman, Davey~Smith, Morris,
  and Davies}{Barry et~al.}{2023}]{barry_how_2023}
Barry, C.-J.~S., V.~M. Walker, R.~Cheesman, G.~Davey~Smith, T.~T. Morris, and
  N.~M. Davies (2023, April).
\newblock How to estimate heritability: a guide for genetic epidemiologists.
\newblock {\em International Journal of Epidemiology\/}~{\em 52\/}(2),
  624--632.

\bibitem[\protect\citeauthoryear{Boyle, Li, and Pritchard}{Boyle
  et~al.}{2017}]{boyle2017expanded}
Boyle, E.~A., Y.~I. Li, and J.~K. Pritchard (2017).
\newblock An expanded view of complex traits: from polygenic to omnigenic.
\newblock {\em Cell\/}~{\em 169\/}(7), 1177--1186.

\bibitem[\protect\citeauthoryear{Brumpton, Sanderson, Heilbron, Hartwig,
  Harrison, Vie, Cho, Howe, Hughes, Boomsma, et~al.}{Brumpton
  et~al.}{2020}]{brumpton2020avoiding}
Brumpton, B., E.~Sanderson, K.~Heilbron, F.~P. Hartwig, S.~Harrison, G.~{\AA}.
  Vie, Y.~Cho, L.~D. Howe, A.~Hughes, D.~I. Boomsma, et~al. (2020).
\newblock Avoiding dynastic, assortative mating, and population stratification
  biases in mendelian randomization through within-family analyses.
\newblock {\em Nature Communications\/}~{\em 11\/}(1), 3519.

\bibitem[\protect\citeauthoryear{Bulik-Sullivan, Loh, Finucane, Ripke, Yang,
  {Schizophrenia Working Group of the Psychiatric Genomics Consortium},
  Patterson, Daly, Price, and Neale}{Bulik-Sullivan
  et~al.}{2015}]{bulik-sullivan_ld_2015}
Bulik-Sullivan, B.~K., P.-R. Loh, H.~K. Finucane, S.~Ripke, J.~Yang,
  {Schizophrenia Working Group of the Psychiatric Genomics Consortium},
  N.~Patterson, M.~J. Daly, A.~L. Price, and B.~M. Neale (2015, March).
\newblock {LD} {Score} regression distinguishes confounding from polygenicity
  in genome-wide association studies.
\newblock {\em Nature Genetics\/}~{\em 47\/}(3), 291--295.

\bibitem[\protect\citeauthoryear{Falconer}{Falconer}{1996}]{falconer1996introduction}
Falconer, D.~S. (1996).
\newblock {\em Introduction to Quantitative Genetics}.
\newblock Pearson Education India.

\bibitem[\protect\citeauthoryear{Fr{\'e}chet}{Fr{\'e}chet}{1951}]{frechet1951tableaux}
Fr{\'e}chet, M. (1951).
\newblock Sur les tableaux de corr{\'e}lation dont les marges sont donn{\'e}es.
\newblock {\em Annals University Lyon: Series A\/}~{\em 14}, 53--77.

\bibitem[\protect\citeauthoryear{Gao and Zhao}{Gao and
  Zhao}{2024}]{gaoCounterfactualExplainabilityBlackbox2024}
Gao, Z. and Q.~Zhao (2024, November).
\newblock Counterfactual explainability of black-box prediction models.

\bibitem[\protect\citeauthoryear{Hartl, Clark, and Clark}{Hartl
  et~al.}{1997}]{hartl1997principles}
Hartl, D.~L., A.~G. Clark, and A.~G. Clark (1997).
\newblock {\em Principles of Population Genetics}, Volume 116.
\newblock Sunderland, MA: Sinauer Associates.

\bibitem[\protect\citeauthoryear{Hernan and Robins}{Hernan and
  Robins}{2023}]{hernanCausalInferenceWhat2023}
Hernan, M.~A. and J.~M. Robins (2023, December).
\newblock {\em Causal {{Inference}}: {{What If}}}.
\newblock CRC Press.

\bibitem[\protect\citeauthoryear{Hoeffding}{Hoeffding}{1940}]{hoeffding1940masstabinvariante}
Hoeffding, W. (1940).
\newblock Masstabinvariante korrelationstheorie.
\newblock {\em Schriften des Mathematischen Instituts und Instituts fur
  Angewandte Mathematik der Universitat Berlin\/}~{\em 5}, 181--233.

\bibitem[\protect\citeauthoryear{Hoeffding}{Hoeffding}{1948}]{hoeffding1948}
Hoeffding, W. (1948).
\newblock {A Class of Statistics with Asymptotically Normal Distribution}.
\newblock {\em The Annals of Mathematical Statistics\/}~{\em 19\/}(3), 293 --
  325.

\bibitem[\protect\citeauthoryear{Holland, Nyquist, and
  Cervantes-Martínez}{Holland et~al.}{2002}]{holland2002}
Holland, J.~B., W.~E. Nyquist, and C.~T. Cervantes-Martínez (2002).
\newblock {\em Estimating and Interpreting Heritability for Plant Breeding: An
  Update}, Chapter~2, pp.\  9--112.
\newblock John Wiley \& Sons, Ltd.

\bibitem[\protect\citeauthoryear{Holland}{Holland}{1986}]{holland1986statistics}
Holland, P.~W. (1986).
\newblock Statistics and causal inference.
\newblock {\em Journal of the American Statistical Association\/}~{\em
  81\/}(396), 945--960.

\bibitem[\protect\citeauthoryear{Hooker}{Hooker}{2007}]{hookerGeneralizedFunctionalANOVA2007}
Hooker, G. (2007).
\newblock Generalized {{Functional ANOVA Diagnostics}} for {{High-Dimensional
  Functions}} of {{Dependent Variables}}.
\newblock {\em Journal of Computational and Graphical Statistics\/}~{\em
  16\/}(3), 709--732.

\bibitem[\protect\citeauthoryear{Imbens and Rubin}{Imbens and
  Rubin}{2015}]{imbens2015causal}
Imbens, G.~W. and D.~B. Rubin (2015).
\newblock {\em Causal Inference in Statistics, Social, and Biomedical
  Sciences}.
\newblock New York: Cambridge University Press.

\bibitem[\protect\citeauthoryear{Jacquard}{Jacquard}{1983}]{jacquard1983heritability}
Jacquard, A. (1983).
\newblock Heritability: one word, three concepts.
\newblock {\em Biometrics\/}~{\em 39}, 465--477.

\bibitem[\protect\citeauthoryear{Jinks and Fulker}{Jinks and
  Fulker}{1970}]{jinks1970comparison}
Jinks, J.~L. and D.~W. Fulker (1970).
\newblock Comparison of the biometrical genetical, mava, and classical
  approaches to the analysis of the human behavior.
\newblock {\em Psychological bulletin\/}~{\em 73\/}(5), 311--349.

\bibitem[\protect\citeauthoryear{Kohler, Behrman, and Schnittker}{Kohler
  et~al.}{2011}]{kohler2011social}
Kohler, H.-P., J.~R. Behrman, and J.~Schnittker (2011).
\newblock Social science methods for twins data: Integrating causality,
  endowments, and heritability.
\newblock {\em Biodemography and Social Biology\/}~{\em 57\/}(1), 88--141.

\bibitem[\protect\citeauthoryear{Kong, Thorleifsson, Frigge, Vilhjalmsson,
  Young, Thorgeirsson, Benonisdottir, Oddsson, Halldorsson, Masson,
  et~al.}{Kong et~al.}{2018}]{kong2018nature}
Kong, A., G.~Thorleifsson, M.~L. Frigge, B.~J. Vilhjalmsson, A.~I. Young, T.~E.
  Thorgeirsson, S.~Benonisdottir, A.~Oddsson, B.~V. Halldorsson, G.~Masson,
  et~al. (2018).
\newblock The nature of nurture: Effects of parental genotypes.
\newblock {\em Science\/}~{\em 359\/}(6374), 424--428.

\bibitem[\protect\citeauthoryear{Manolio, Collins, Cox, Goldstein, Hindorff,
  Hunter, McCarthy, Ramos, Cardon, Chakravarti, et~al.}{Manolio
  et~al.}{2009}]{manolio2009finding}
Manolio, T.~A., F.~S. Collins, N.~J. Cox, D.~B. Goldstein, L.~A. Hindorff,
  D.~J. Hunter, M.~I. McCarthy, E.~M. Ramos, L.~R. Cardon, A.~Chakravarti,
  et~al. (2009).
\newblock Finding the missing heritability of complex diseases.
\newblock {\em Nature\/}~{\em 461\/}(7265), 747--753.

\bibitem[\protect\citeauthoryear{Neale and Cardon}{Neale and
  Cardon}{1992}]{neale1992methodology}
Neale, M.~C. and L.~R. Cardon (1992).
\newblock {\em Methodology for Genetic Studies of Twins and Families\/} (1
  ed.), Volume~67 of {\em NATO Science Series D:}.
\newblock Springer Dordrecht.
\newblock Part of the NATO Science Series D: (ASID, volume 67); Topics: Human
  Genetics, Methodology of the Social Sciences.

\bibitem[\protect\citeauthoryear{Neale, Hunter, Pritikin, Zahery, Brick,
  Kirkpatrick, Estabrook, Bates, Maes, and Boker}{Neale
  et~al.}{2016}]{neale2016openmx}
Neale, M.~C., M.~D. Hunter, J.~N. Pritikin, M.~Zahery, T.~R. Brick, R.~M.
  Kirkpatrick, R.~Estabrook, T.~C. Bates, H.~H. Maes, and S.~M. Boker (2016).
\newblock Open{M}x 2.0: {E}xtended structural equation and statistical
  modeling.
\newblock {\em Psychometrika\/}~{\em 81\/}(2), 535--549.

\bibitem[\protect\citeauthoryear{Pearl}{Pearl}{2009}]{pearl2009}
Pearl, J. (2009).
\newblock {\em Causality: {{Models}}, Reasoning, and Inference\/} (2 ed.).
\newblock New York: Cambridge University Press.

\bibitem[\protect\citeauthoryear{Pearl and Mackenzie}{Pearl and
  Mackenzie}{2018}]{pearl2018book}
Pearl, J. and D.~Mackenzie (2018).
\newblock {\em The Book of Why: The New Science of Cause and Effect}.
\newblock Basic books.

\bibitem[\protect\citeauthoryear{Piepho and Mohring}{Piepho and
  Mohring}{2007}]{piepho2007computing}
Piepho, H.-P. and J.~Mohring (2007).
\newblock Computing heritability and selection response from unbalanced plant
  breeding trials.
\newblock {\em Genetics\/}~{\em 177\/}(3), 1881--1888.

\bibitem[\protect\citeauthoryear{Richardson and Robins}{Richardson and
  Robins}{2013}]{richardson2013single}
Richardson, T.~S. and J.~M. Robins (2013).
\newblock Single world intervention graphs (swigs): A unification of the
  counterfactual and graphical approaches to causality.
\newblock {\em Center for the Statistics and the Social Sciences, University of
  Washington Series. Working Paper\/}~{\em 128\/}(30), 2013.

\bibitem[\protect\citeauthoryear{Rosenbaum and Rubin}{Rosenbaum and
  Rubin}{1983}]{rosenbaum1983central}
Rosenbaum, P.~R. and D.~B. Rubin (1983).
\newblock The central role of the propensity score in observational studies for
  causal effects.
\newblock {\em Biometrika\/}~{\em 70\/}(1), 41--55.

\bibitem[\protect\citeauthoryear{Rubin}{Rubin}{1980}]{rubin1980randomization}
Rubin, D.~B. (1980).
\newblock Randomization analysis of experimental data: The fisher randomization
  test comment.
\newblock {\em Journal of the American Statistical Association\/}~{\em
  75\/}(371), 591--593.

\bibitem[\protect\citeauthoryear{Schmidt, Hartung, Bennewitz, and
  Piepho}{Schmidt et~al.}{2019}]{schmidt2019heritability}
Schmidt, P., J.~Hartung, J.~Bennewitz, and H.-P. Piepho (2019).
\newblock Heritability in plant breeding on a genotype-difference basis.
\newblock {\em Genetics\/}~{\em 212\/}(4), 991--1008.

\bibitem[\protect\citeauthoryear{Tudball, Smith, and Zhao}{Tudball
  et~al.}{2022}]{tudball2022almost}
Tudball, M.~J., G.~D. Smith, and Q.~Zhao (2022).
\newblock Almost exact mendelian randomization.
\newblock arXiv:2208.14035.

\bibitem[\protect\citeauthoryear{Veller and Coop}{Veller and
  Coop}{2024}]{veller2024interpreting}
Veller, C. and G.~M. Coop (2024).
\newblock Interpreting population-and family-based genome-wide association
  studies in the presence of confounding.
\newblock {\em PLoS Biology\/}~{\em 22\/}(4), e3002511.

\bibitem[\protect\citeauthoryear{Visscher, Medland, Ferreira, Morley, Zhu,
  Cornes, Montgomery, and Martin}{Visscher
  et~al.}{2006}]{visscher2006assumption}
Visscher, P.~M., S.~E. Medland, M.~A.~R. Ferreira, K.~I. Morley, G.~Zhu, B.~K.
  Cornes, G.~W. Montgomery, and N.~G. Martin (2006).
\newblock Assumption-free estimation of heritability from genome-wide
  identity-by-descent sharing between full siblings.
\newblock {\em PLoS Genetics\/}~{\em 2\/}(3), e41.

\bibitem[\protect\citeauthoryear{Wang}{Wang}{2015}]{wang2015current}
Wang, R. (2015).
\newblock Current open questions in complete mixability.
\newblock {\em Probability Surveys\/}~{\em 12}, 13--32.

\bibitem[\protect\citeauthoryear{Wright}{Wright}{1920}]{wright1920relative}
Wright, S. (1920).
\newblock The relative importance of heredity and environment in determining
  the piebald pattern of guinea-pigs.
\newblock {\em Proceedings of the National Academy of Sciences\/}~{\em 6\/}(6),
  320--332.

\bibitem[\protect\citeauthoryear{Xiao and Yao}{Xiao and
  Yao}{2020}]{xiao2020note}
Xiao, Y. and J.~Yao (2020).
\newblock A note on joint mix random vectors.
\newblock {\em Communications in Statistics-Theory and Methods\/}~{\em
  49\/}(12), 3063--3072.

\bibitem[\protect\citeauthoryear{Yang, Benyamin, McEvoy, Gordon, Henders,
  Nyholt, Madden, Heath, Martin, Montgomery, Goddard, and Visscher}{Yang
  et~al.}{2010}]{yang_common_2010}
Yang, J., B.~Benyamin, B.~P. McEvoy, S.~Gordon, A.~K. Henders, D.~R. Nyholt,
  P.~A. Madden, A.~C. Heath, N.~G. Martin, G.~W. Montgomery, M.~E. Goddard, and
  P.~M. Visscher (2010, July).
\newblock Common {SNPs} explain a large proportion of the heritability for
  human height.
\newblock {\em Nature Genetics\/}~{\em 42\/}(7), 565--569.
\newblock Publisher: Nature Publishing Group.

\bibitem[\protect\citeauthoryear{Yang, Zeng, Goddard, Wray, and Visscher}{Yang
  et~al.}{2017}]{yang_concepts_2017}
Yang, J., J.~Zeng, M.~E. Goddard, N.~R. Wray, and P.~M. Visscher (2017,
  September).
\newblock Concepts, estimation and interpretation of {SNP}-based heritability.
\newblock {\em Nature Genetics\/}~{\em 49\/}(9), 1304--1310.
\newblock Publisher: Nature Publishing Group.

\bibitem[\protect\citeauthoryear{Young}{Young}{2019}]{young2019solving}
Young, A.~I. (2019).
\newblock Solving the missing heritability problem.
\newblock {\em PLoS Genetics\/}~{\em 15\/}(6), e1008222.

\bibitem[\protect\citeauthoryear{Young, Frigge, Gudbjartsson, Thorleifsson,
  Bjornsdottir, Sulem, Masson, Thorsteinsdottir, Stefansson, and Kong}{Young
  et~al.}{2018}]{young2018relatedness}
Young, A.~I., M.~L. Frigge, D.~F. Gudbjartsson, G.~Thorleifsson,
  G.~Bjornsdottir, P.~Sulem, G.~Masson, U.~Thorsteinsdottir, K.~Stefansson, and
  A.~Kong (2018).
\newblock Relatedness disequilibrium regression estimates heritability without
  environmental bias.
\newblock {\em Nature Genetics\/}~{\em 50\/}(9), 1304--1310.

\bibitem[\protect\citeauthoryear{Young, Nehzati, Benonisdottir, Okbay,
  Jayashankar, Lee, Cesarini, Benjamin, Turley, and Kong}{Young
  et~al.}{2022}]{young2022mendelian}
Young, A.~I., S.~M. Nehzati, S.~Benonisdottir, A.~Okbay, H.~Jayashankar,
  C.~Lee, D.~Cesarini, D.~J. Benjamin, P.~Turley, and A.~Kong (2022).
\newblock Mendelian imputation of parental genotypes improves estimates of
  direct genetic effects.
\newblock {\em Nature Genetics\/}~{\em 54\/}(6), 897--905.

\bibitem[\protect\citeauthoryear{Young}{Young}{2024}]{young2024genome}
Young, A.~S. (2024).
\newblock Genome-wide association studies have problems due to confounding: Are
  family-based designs the answer?
\newblock {\em PLoS Biology\/}~{\em 22\/}(4), e3002568.

\bibitem[\protect\citeauthoryear{Zaitlen and Kraft}{Zaitlen and
  Kraft}{2012}]{zaitlen2012heritability}
Zaitlen, N. and P.~Kraft (2012).
\newblock Heritability in the genome-wide association era.
\newblock {\em Human Genetics\/}~{\em 131}, 1655--1664.

\bibitem[\protect\citeauthoryear{Zuk, Hechter, Sunyaev, and Lander}{Zuk
  et~al.}{2012}]{zuk2012mystery}
Zuk, O., E.~Hechter, S.~R. Sunyaev, and E.~S. Lander (2012).
\newblock The mystery of missing heritability: Genetic interactions create
  phantom heritability.
\newblock {\em Proceedings of the National Academy of Sciences\/}~{\em
  109\/}(4), 1193--1198.

\end{thebibliography}
\end{document}